\newif\iftacversion
\tacversionfalse 

\newif\ifincludenotation
\includenotationtrue 

\newif\ifincludealggeoremark
\includealggeoremarkfalse 

\newif\ifremarktaskspace
\remarktaskspacefalse 

\newif\iftransportfig
\transportfigtrue 

\newif\ifincludeappendices
\includeappendicestrue 

\newif\ifincludechristoffel
\includechristoffeltrue 

\newif\ifincludeconstruction
\includeconstructionfalse 

\newif\ifincludeexapndedBD
\includeexapndedBDfalse 

\newif\ifincludethreeform
\includethreeformfalse 

\newif\ifshortproof
\shortprooftrue 

\iftacversion
\documentclass[journal,twoside,web]{ieeecolor}
\usepackage{generic}
\usepackage{cite}
\usepackage{amsmath,amssymb,amsfonts}
\usepackage{algorithmic}
\usepackage{graphicx}
\usepackage{algorithm,algorithmic}
\usepackage{hyperref}
\hypersetup{hidelinks=true}
\usepackage{textcomp}
\else
\documentclass[letterpaper, 10 pt]{IEEEtran}
\fi

\usepackage{cite}
\usepackage{hyperref}
 \usepackage{algorithmic}
\usepackage{amsmath}
\usepackage{amssymb}
\usepackage{amstext}
\usepackage{color}
\usepackage{url}
\renewcommand{\vec}[1]{ {\mathbf{#1}}}
\usepackage{multirow}
\usepackage{bm}

\usepackage{algorithm}
\usepackage{algorithmic}

\usepackage{setspace}

\newcommand{\T}{^\top}

\newcommand \se     {\mathrm{se}}

\newcommand{\bzero}{\mbox{\boldmath $0$}}

\usepackage{mdframed}

{\end{tabbing} \end{mdframed} \vspace{5px}}

\newcommand{\vI}{{\mbox{\boldmath$\mathit{I}$}}}

\newcommand{\BM}{\begin{bmatrix}}
\newcommand{\EM}{\end{bmatrix}}

\newcommand{\bPhi}{{\bf \Phi}}

\newcommand{\bone}{{\bf 1}}

\newcommand{\va}{\vec{a}}
\newcommand{\vB}{\vec{B}}
\newcommand{\vb}{\vec{b}}

\newcommand{\dd}[2]{\frac{{\rm d}#1}{{\rm d}#2}}

\newcommand{\ddt}{\dd{}{t}}

\newcommand{\vh}{\vec{h}}

\newcommand{\vJ}{\vec{J}}

\newcommand{\vw}{\vec{w}}

\newcommand{\qd}{\dot{\q}}

\newcommand{\vv}{\vec{v}}

\newcommand{\vX}{\vec{X}}

\newcommand{\XM}[2]{{}^{#1}\vX_{#2}}
\newcommand{\XMT}[2]{{}^{#1}\vX_{#2}\T}

\newcommand{\beq}{\begin{equation}}
\newcommand{\eeq}{\end{equation} }

\newcommand{\f}{\mathbf{f}}

\newcommand{\delt}[2][]{\ifthenelse{\equal{#1}{}} { \delta \mathbf{#2}}{\delta \mathbf{#2}_{#1}}}

\newcommand{\R}{\mathbb{R}}

\newcommand{\X}{\mathfrak{X}}

\usepackage{graphicx}
\usepackage{empheq}
\usepackage{amsthm}
\usepackage{bbding}

\usepackage{mathrsfs}

\renewcommand{\se}[1]{\mathsf{se}(#1)}

\renewcommand{\l}{\langle \! \langle}
\renewcommand{\r}{\rangle \! \rangle}

\newcommand{\Q}{\mathcal{Q}}
\newcommand{\Qb}{\overline{\mathcal{Q}}}
\newcommand{\nablab}{\overline{\nabla}}
\renewcommand{\qd}{\dot{q}}

\newcommand{\Gb}{\overline{G}}
\newcommand{\Cb}{\overline{C}}

\newcommand{\Gammab}{\overline{\Gamma}{}}

\renewcommand{\R}{\mathbb{R}}
\usepackage{scalerel}

\renewcommand{\qd}{\dot{q}}
\renewcommand{\L}{\mathscr{L}}
\newcommand{\Lie}[2]{ \L_{#1}{#2} }
\usepackage{multicol}
\usepackage{amssymb}
\usepackage{stackengine}
\usepackage{scalerel}
\usepackage{xcolor}
\usepackage{graphicx}
\let\labelindent\undefined
\usepackage{enumitem}
\usepackage{booktabs}

\newcommand{\G}{\mathbb{G}}
\renewcommand{\Gb}{\overline{\G}}
\renewcommand{\vb}{\overline{v}}
\newcommand{\Hb}{\overline{H}}
\renewcommand{\Cb}{\overline{C}}
\newcommand{\Xb}{\overline{X}}
\DeclareMathOperator*{\argmin}{argmin}
\newcommand{\gh}{\mathsf{h}}

\newtheorem{theorem}{Theorem}[section]

\newtheorem{proposition}{Proposition}[section]
\newtheorem{remark}{Remark}
\newtheorem{corollary}{Corollary}[section]
\newcommand{\addspace}[1]{}

\newcommand{\later}[1]{}

\newcommand\openbigstar[1][0.7]{%
  \scalebox{.5}{\scalerel*{%
    \stackinset{c}{-.125pt}{c}{}{\scalebox{#1}{\color{white}{$\bigstar$}}}{%
      $\bigstar$}%
  }{\bigstar}}
}
\renewcommand{\star}{{\openbigstar[.5]}}

\title{\Large \bf Coriolis Factorizations and their {\em Connections} to Riemannian Geometry\\[-1ex]}

\author{ Patrick M.~Wensing$^{1}$
and Jean-Jacques E.~Slotine$^{2}$\\[-4ex]~
       \thanks{$^{1}$ Dept.~of Aero.~\& Mech.~Engineering, University of Notre Dame, IN-46556, USA. {\tt pwensing@nd.edu}
       }
       \thanks{$^{2}$ Dept.~of Mech.~Eng., Dept.~of Brain and Cognitive Sciences, and Nonlinear
Systems Laboratory, Massachusetts Institute of Technology, Cambridge, MA, USA 02139 {\tt jjs@mit.edu}
       } 
       }

\begin{document}

\maketitle
\begin{abstract}
Many energy-based control strategies for mechanical systems require the choice of a Coriolis factorization satisfying a skew-symmetry property. This paper (a) explores if and when a control designer has flexibility in this choice, (b) develops a canonical choice related to the Christoffel symbols, and (c) describes how to efficiently perform control computations with it for constrained mechanical systems. We link the choice of a Coriolis factorization to the notion of an affine connection on the configuration manifold and show how properties of the connection relate with the associated factorization. In particular, 
the factorization based on the Christoffel symbols is linked with a torsion-free property that can limit the twisting of system trajectories during passivity-based control. We then develop a way to induce Coriolis factorizations for constrained mechanisms from unconstrained ones, which provides a pathway to use the theory for efficient control computations with high-dimensional systems such as humanoids and quadruped robots with open- and closed-chain mechanisms. 
\ifincludeappendices
A collection of algorithms is provided (and made available open source) to support the recursive computation of passivity-based control laws, adaptation laws, and regressor matrices in future applications.  
\else 
\fi

\end{abstract}

\section{Introduction}

When studying the control of mechanical systems, a common starting point is the system dynamics equation:
\begin{equation}
H(q) \dot{v} + c(q,v) + g(q) = \tau
\label{eq:eom}
\end{equation}
where $q \in \Q$ represents the system configuration with  $\Q$ the configuration manifold, $v\in \R^n$ the generalized velocities, $H(q) \in \mathbb{R}^{n\times n}$ the mass matrix, $c(q,v) \in \R^n$ the Coriolis and centripetal terms, $g(q) \in \mathbb{R}^n$ the generalized gravity force, and $\tau \in \mathbb{R}^n$ the vector of generalized applied forces. 
The Coriolis and centripetal terms depend quadratically on the velocity variables and can be factored as:
\begin{equation}
C(q,v)v = c(q,v)
\label{eq:factorization}
\end{equation}
where $C(q,v) \in \R^{n\times n}$ is linear in $v$ and is referred to as a Coriolis matrix or Coriolis factorization. This paper investigates the properties of such factorizations through the lens of Riemannian geometry (Fig.~\ref{fig:summary}) \cite{do1992riemannian}, and develops methods for computing factorizations with desired properties.


Many factorizations of $c(q,v)$ are possible. For example \cite{bjerkeng2012new}, if $v = [v^1, v^2]\T$ and $c(q,v) = [v^1 v^2 , 0]\T$, valid choices include:
\[
C_a =  \scalebox{.8}{$\begin{bmatrix} v^2 & 0  \\ 0 &  0 \end{bmatrix}$} \quad  C_b = \scalebox{.8}{$\begin{bmatrix} 0 & v^1  \\ 0 &  0 \end{bmatrix}$} 
\]
or any combination of the form $\alpha C_a + (1-\alpha) C_b$ ($\alpha \in \R$).
In applications involving passivity-based control, an additional condition is often imposed:
\begin{equation}
\dot{H} = C+C\T
\label{eq:skew}
\end{equation}
which is equivalent to $\dot{H} -2 C$ being skew-symmetric. When working with generalized coordinates $(q^1, \ldots, q^n)$ with $v^i = \dot{q}^i$, a canonical choice satisfying \eqref{eq:skew} is given by:
\begin{equation}
C_{ij}^\star = \Gamma_{ijk}^\star \dot{q}^k
\label{eq:christoffel_C}
\end{equation}
where Einstein summation (over $k$) convention is adopted and 
\begin{equation}
\Gamma_{ijk}^\star =  \scalebox{1.15}{$\frac{1}{2} \left[ \frac{\partial H_{ij}}{\partial q^k} + \frac{\partial H_{ik}}{\partial q^j}  - \frac{\partial H_{jk}}{\partial q^i} \right]$}
\label{eq:christoffel}
\end{equation}
 are the Christoffel symbols of the first kind. Earlier works \cite{slotine1987adaptive,ortega1989adaptive} often used this factorization, with some suggesting its uniqueness in satisfying \eqref{eq:factorization} and \eqref{eq:skew} \cite{Siciliano}. 
 These observations raise fundamental questions that we undertake herein:

 \vspace{5px}

 \begin{figure}
    \centering
    \includegraphics[width=\columnwidth]{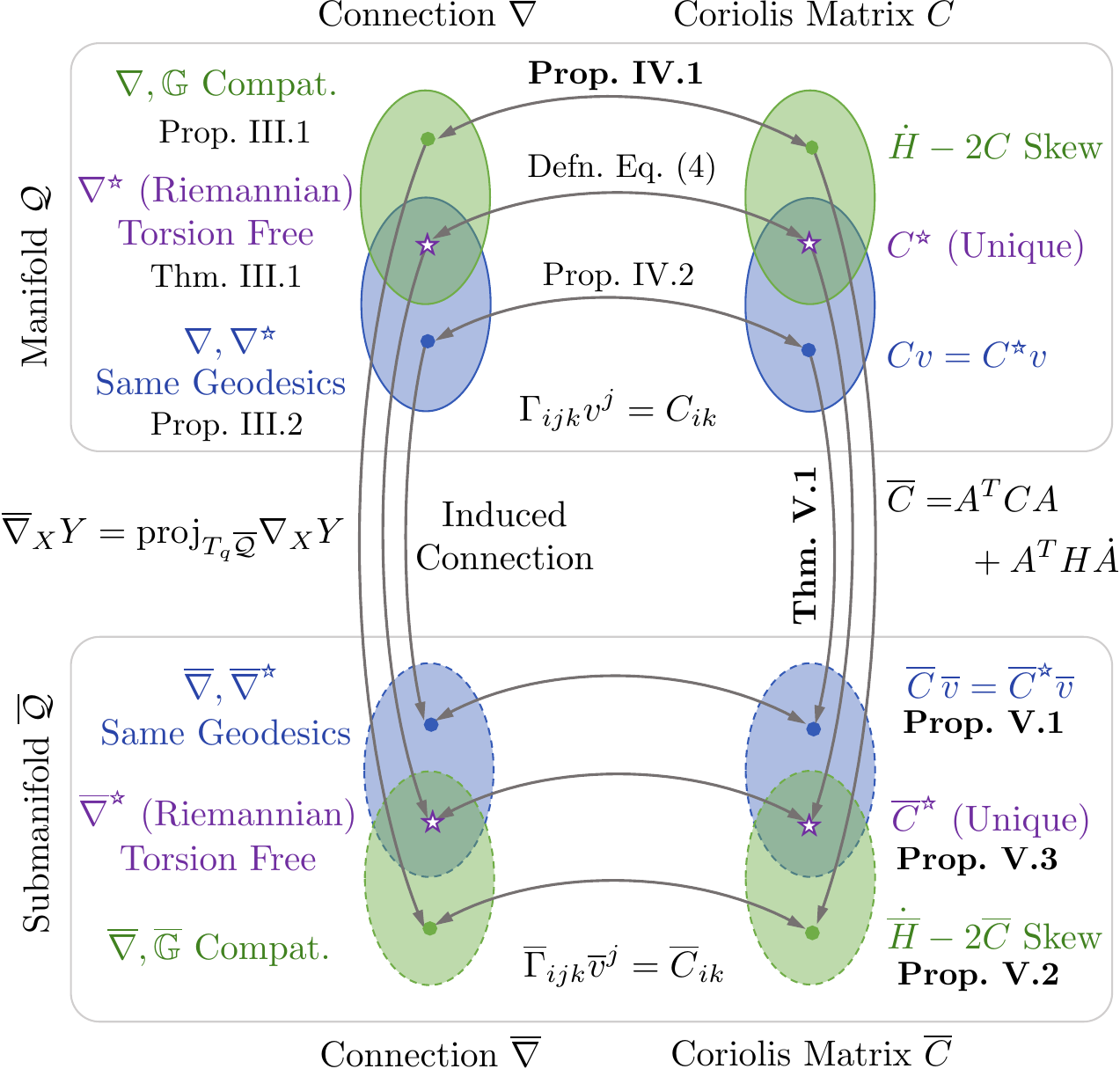}
    \caption{Roadmap. For a mechanical system on a Riemannian manifold $\mathcal{Q}$ with kinetic energy metric $\mathbb{G}$ we study how the properties of any chosen affine connection $\nabla$ relate with those of a corresponding Coriolis matrix $C$ and mass matrix $H$ w.r.t. a choice of generalized velocities $v$. We show how these properties for the system on $\mathcal{Q}$ can be related to a constrained system that evolves on a submanifold $\overline{\mathcal{Q}}$ where the dynamics are described with generalized velocities $\overline{v}$ that are related to $v$ by $v = A \overline{v}$. 
    Bolded results indicate the main contributions.
    }
    \label{fig:summary}
\end{figure}

\begin{enumerate}[label={\bf (Q\arabic*)}, ref=(Q\arabic*), itemsep=2pt]

\item Given the relationship between  \eqref{eq:christoffel_C} and the Christoffel symbols, is there a deeper relationship between other Coriolis factorizations and geometric structures on $\Q$?

\item When do \eqref{eq:factorization} and \eqref{eq:skew} uniquely determine $C$? 


\item How does the canonical factorization \eqref{eq:christoffel_C} extend to generalized velocities?

\item How can such factorizations be practically computed for systems with high degrees of freedom?

\end{enumerate}
%
%
%
%
%
%
%
%

\noindent To summarize the results of the paper, we provide short answers to these questions below, leaving a self-contained introduction of all relevant nomenclature to the main body:

\begin{enumerate}[label={\bf (A\arabic*)}, ref=(Q\arabic*), itemsep=2pt]

\item Every factorization can be associated with an affine connection on $\Q$, wherein geometric properties of the connection are linked with \eqref{eq:factorization} and \eqref{eq:skew} (Fig.~\ref{fig:summary}).

\item Factorizations satisfying \eqref{eq:factorization} and \eqref{eq:skew} are uniquely defined for $n=1,2$. When $n\ge 3$, there are infinitely many such factorizations. For example, for a point mass with $\dot{p}= v \in \R^3$, Newton's equation $f = m \dot{v}$ admits factorizations $C(p,v)= \beta(p) \, (v\times)$ for any 
$\beta:\R^3\rightarrow \R$. 

\item  The canonical factorization \eqref{eq:christoffel_C} generalizes using the Riemannian connection on non-coordinate basis vector fields, ensuring a torsion-free 
property that can limit the twisting of trajectories during passivity-based control.
 
\item We propose a structure-preserving approach to induce factorizations for constrained systems  
 (Fig.~\ref{fig:summary}). By using the Riemannian connection for unconstrained rigid bodies, these results enable forming $C^\star$ from \eqref{eq:christoffel_C} for general constrained mechanical systems. 
 The development provides broader theoretical backing to a previous algorithm \cite[Alg.~1]{echeandia2021numerical} for computing the Coriolis matrix, which is shown to apply 
  to systems with local kinematic loops.
\end{enumerate}

The rest of the document fully develops these answers. 

\section{Related Work}
\label{sec:related_work}

We begin by taking a step back and considering how the structure of the Coriolis terms in \eqref{eq:eom} have influenced dynamics calculation methods and the control of mechanical systems.  The first numerical methods for computing the left side of \eqref{eq:eom} were Recursive Newton-Euler (RNE) approaches, appearing in the late 1970s, and motivated by dynamic analysis for walking machines \cite{orin1979kinematic,stepanenko1976dynamics}. Developments in the early 1980s soon after applied these same methods for computed-torque control \cite{luhwalkerpaul80} of robot manipulators. These and other computed-torque control laws require evaluating the Coriolis terms $c(q,v)$ (so they can be canceled out via feedback linearization), which can be accomplished via RNE methods in $O(n)$ complexity.

The dependency of computed-torque controllers on accurate dynamic models for exact cancellations pushed subsequent developments of energy-based methods \cite{slotine1988putting}. The original PD+ control law of Takagaki and Arimoto \cite{Arimoto81}, was an early example, and has served as motivation for passivity-based and adaptive control from the mid-1980s to the present day \cite{koditschek1984natural,slotine1987adaptive, ortega1989adaptive, ortega2013passivity, pucci2015collocated, dietrich2016passive, chopra2022passivity}. In many of these developments, a factorization is used to avoid exactly canceling Coriolis terms, where the satisfaction of the skew-symmetry property \eqref{eq:skew} for the Coriolis matrix plays a critical role in many Lyapunov arguments.  Passivity-based frameworks are also central for the stability of teleoperation with time delays (e.g., the work starting with \cite{anderson1989bilateral, niemeyer1991stable}). Throughout these works and many robotics texts (e.g., \cite{Siciliano, lynch2017modern}), the factorization based on Christoffel symbols, originally appearing in \cite{slotine1987adaptive}, has remained commonly adopted.

Within passivity-based control, one often requires computing the product $C(q,v) v_r$ where $v_r$ is some reference velocity. As a result, 
many authors \cite{Niemeyer90, niemeyer1991performance, lin1995skew, ploen1999skew} developed variants of the RNE algorithm that were capable of computing $C(q,v) v_r$ in place of $C(q,v)v$ to implement passivity-based controllers efficiently. Care was taken to ensure that the recursive computations were compatible with the skew property \eqref{eq:skew}, with \cite{Niemeyer90} providing the first RNE method compatible with the Christoffel-consistent factorization (for systems with revolute and prismatic joints). Later, \cite{lin1995skew} provided an explicit construction for the non-uniqueness of factorizations satisfying the skew-symmetry property for general mechanisms. 

While the majority of these developments were focused on manipulators, parallel advances considered Coriolis factorizations for underwater systems \cite{fossen1991nonlinear,schjolberg1996modeling} and passivity-based control of free-floating space manipulators \cite{wang2009passivity, wang2012recursive}. As a common thread, these lines of work considered Coriolis factorizations for unconstrained rigid bodies and related their properties to those of a constrained system as a whole. This approach motivates our strategy in Sec.~\ref{sec:computation}. Our contributions generalize these past results, showing, in hindsight, that \cite{Niemeyer90} provided a skew-symmetric RNE method compatible with the Christoffel-consistent factorization for a broader class of mechanisms than originally claimed.

The transpose of a Coriolis matrix (satisfying \eqref{eq:skew}) also appears in the Hamiltonian formulation of the dynamics \cite[Eq.~6]{DeLuca06} in the form $C(q,v)\T v$. As a result, the Coriolis matrix can be used to implement disturbance observers based on monitoring the generalized momentum \cite{DeLuca06, bledt2018contact}. This observation motivated numerical methods to compute $C$ explicitly \cite{DeLuca09,echeandia2021numerical}, the former of which has a closely related variant that was independently developed as part of the popular C/C++ Pinocchio dynamics package \cite{carpentier2019pinocchio}. In contact detection, any freedom in choosing $C$ is moot since the uniqueness of $\dot{H} v$ and $C(q,v)v$ implies that $C(q,v)\T v$ will result in the same product for any factorization satisfying \eqref{eq:skew}. An efficient method for computing this product is provided in \cite{wang2013recursive} for open-chain systems. The general results herein provide computational tools to extend disturbance-observer computations to general constrained mechanisms.

More recently, passivity-based versions \cite{henze2016passivity,englsberger2020mptc,kurtz2020approximate} of whole-body quadratic-programming-based controllers (e.g., \cite{escande2014hierarchical,kuindersma2016optimization}) have been pursued under the same robustness motivations as for original passivity-based manipulator control laws. In some cases, the configuration-space Coriolis matrix $C$ must be explicitly computed (e.g., in \cite{englsberger2020mptc}, using methods from \cite{garofalo2013closed}) to construct task-space target dynamics with desired passivity properties.  
This paper informs a canonical choice of factorization $C$ for these works and outlines efficient methods for computing it in free-floating systems described by generalized speeds.

As mentioned in the introduction, the main machinery for our developments comes from recognizing a relationship between factorizations of the Coriolis matrix and affine connections on the configuration manifold.
Given these links, the methods in the paper provide an immediate computational approach for calculating covariant derivatives \cite[Chap.~2]{do1992riemannian} associated with complex mechanical systems. This result could have bearing on geometry-based optimization strategies (e.g., \cite{boutselis2020discrete}), geometry-based observers \cite{mishra2022reduced}, or computational workflows to assess geometric properties such as differential flatness \cite{welde2022role} or nonlinear controllability \cite{lewis1997configuration}. More broadly, there is a growing interest in geometric methods for robotics \cite{jaquier2022riemannian}, which suggests a potential for other impacts by linking fundamental robot modeling choices more deeply with geometric properties.

Beyond the above robot control work, the paper also builds on foundational work on geometric mechanics. In particular, our answer to Q3 is motivated by the generalization of the Christoffel symbols from \cite{bullo2004geometric}.  The relationship noted between passivity-based control laws and affine connections in \cite{van2017} and \cite{reyes2019virtual} also provides a starting point for one of our propositions. We make the exact connections with this past work precise throughout. 

Overall, we contribute a comprehensive synthesis relating the properties of Coriolis matrices with affine connections, and detailing how these relationships can be used for computing Coriolis matrices with desired properties for complex mechanical systems described using generalized velocities.

The remainder of the document is organized as follows. In Sec.~\ref{sec:background} we introduce key concepts from differential geometry through the perspective of mechanics, including an introduction to affine connections and their properties. We then establish a link between the choice of an affine connection and the choice of a Coriolis factorization in Sec.~\ref{sec:link} (answering Q1-Q3 from the intro) and demonstrate the effect of factorization choice on passivity-based control. Sec.~\ref{sec:induce} describes general results showing how one can induce factorizations for a mechanical system evolving on a submanifold, while Sec.~\ref{sec:computation} takes the specific route of applying those results to constraining a set of rigid bodies (answering Q4), with a focus on computational considerations. 
\ifincludenotation
Table \ref{tab:notation} provides a summary of key notation across this development. 
\else
\fi
Concluding remarks are given Sec.~\ref{sec:conclude}.

\ifincludenotation

\begin{table}[h]
\centering
\small
\caption{Summary of Select Notation}
\begin{tabular}{ll}
\toprule
\textbf{Symbol} & \textbf{Description} \\
\midrule
Geometry \hspace{-240 px}\\ \midrule
$ \mathcal{Q}$          & Configuration manifold                   \\
$\X(\mathcal{Q})$       & Set of smooth vector fields over $\mathcal{Q}$ \\
$ \mathbb{G} $          & Riemannian metric on $ \mathcal{Q} $                               \\
$ \l \cdot, \cdot \r_q $ & Inner product on $ T_q\mathcal{Q} $, induced by $ \mathbb{G} $ \\
$ \Lie{X}{f} $          & Lie derivative of $ f $ along $ X $ \\
$[ \cdot, \cdot ] $     & Lie bracket of vector fields \\
$ \nabla $              & Generic affine connection                     \\
$ \nabla^\star $              & Riemannian connection                     \\
$D$ & Contorsion tensor \\
$ \Gamma_{ijk}^\star $  & Christoffel symbols of the first kind     \\
$ \overline{\nabla} $              & Induced affine connection                   \\
\midrule
\multicolumn{2}{l}{Mechanical Systems} \hspace{-240 px}
\\ \midrule
$ q \in \mathcal{Q} $   & Configuration    \\
$ T_q\mathcal{Q} $      & Tangent space at point $ q $ on $ \mathcal{Q} $                               \\
$ v \in \mathbb{R}^n $  & Generalized velocity vector                                                   \\
$ \dot{q} \in T_q\mathcal{Q} $             & Time derivative of the configuration            \\
$ X_i\in \X(\mathcal{Q}) $                & $i$-th basis vector field on $ \mathcal{Q} $ \\
$ H(q) \in \mathbb{R}^{n \times n} $ & Mass (or inertia) matrix                                   \\
$ C(q, v) \in \mathbb{R}^{n \times n} $ & Coriolis matrix                              \\
$ g(q) \in \mathbb{R}^n $ & Generalized gravitational forces                                             \\
$ \tau \in \mathbb{R}^n $ & Generalized applied forces or torques                                        \\
$ p = H(q) v \in \mathbb{R}^n $             & Generalized momentum \\
$ C_{ij}^\star $        & Canonical Coriolis matrix via $ \Gamma_{ikj}^\star \dot{q}^k $   \\
\midrule
\multicolumn{2}{l}{Constrained Mechanical Systems}\hspace{-240 px}\\ \midrule
$ \overline{\mathcal{Q}} \subseteq \mathcal{Q}$          & Constrained configuration manifold                  \\
$ \overline{X}_i \in \X(\overline{\mathcal{Q}}) $                & $i$-th basis vector field on $ \mathcal{\overline{Q}} $ \\
$ \overline{v} \in \mathbb{R}^m $  & Generalized velocity vector                                                 \\
$ A(q) \in \mathbb{R}^{n \times m}$                & Transformation matrix: $v = A \overline{v}$ \\
$ \overline{H}(q) \in \mathbb{R}^{m \times m} $ & Induced mass matrix                                  \\
$ \overline{C}(q, \overline{v}) \in \mathbb{R}^{m \times m} $ & Induced Coriolis matrix                             \\
\midrule
Multibody \hspace{-240 px}\\ \midrule
$\mathbf{v}_i$                   & Spatial velocity of $i$-th body \\
$\mathsf{v}_i$                   & Stacked spatial velocities for $i$-th cluster \\
${\tt v}_i / \overline{{\tt v}}_i$                  & Generalized velocities for the $i$-th joint \\
$[\cdot]^\wedge$                & Hat map from $\mathbb{R}^6$ to $\se{3}$ \\    
$\XM{i}{j}$                     & Spatial transform from body $j$ to $i$ \\
$\bPhi_j$                       & Joint mode matrix for joint $i$ \\ 
$\theta_i$             & Inertial parameters for body/cluster $i$ \\
\bottomrule
\end{tabular}
\label{tab:notation}
\end{table}

\else
\fi

\section{Background}
\label{sec:background}

The goal of this section is to build intuition into differential geometry by bringing physical insights to the concepts of covariant derivatives and affine connections. Readers already well-versed in these topics may choose to skip subsections \ref{sec:mechanical_systems} and \ref{sec:riemannian_geometry}, though we hope they would still discover a few new insights along the way.

\subsection{Mechanical Systems}
\label{sec:mechanical_systems}
We consider a mechanical system with a smooth $n$-dimensional configuration manifold $\Q$. Denoting by $\mathfrak{X}(\Q)$ the set of smooth vector fields on $\Q$, we consider a collection of vector fields $X_1,\ldots,X_n\in \X(\Q)$ so that at each point $q \in \Q$ \[
\text{span}(X_1(q), \ldots, X_n(q)) = T_q \Q
\]
where $T_q \Q$ gives the tangent space to $\Q$ at $q$. We adopt generalized velocities $v\in \R^n$ to describe the dynamics of the system, such that at each $q\in Q$ we write $\qd \in T_q \Q$ as $\qd = X_i v^i$ with Einstein summation convention adopted.  The potential energy $V(q)$ of the system is purely configuration dependent, while the kinetic energy can be written as $T(q,v) = \frac{1}{2} v\T H(q) v$ where $H(q)\in \mathbb{R}^{n\times n}$ is the symmetric positive definite mass matrix from \eqref{eq:eom}. 

Via any variety of methods (e.g., Hamel's equations, Kane's method, Gibbs Appel, etc.), the dynamics of the system are known to take the form \eqref{eq:eom}. 
The components of the generalized conservative forces are given by $g_i(q) = \Lie{X_i}{V}(q)$  where $\Lie{X_i}{V}$ denotes the Lie derivative of $V$ along $X_i$ \cite[Ch.~7, Sec.~2.1]{murray2017mathematical}. 

\begin{remark}
If we write our Lagrangian as $L(q,v) = \frac{1}{2} v\T H(q) v-V(q)$, then the equations of motion given by Hamel's equations \cite{Hamel24,ball2012variational,muller2023hamel} are:
\begin{equation}
\tau_i = \frac{{\rm d}}{{\rm d}t} \frac{\partial L}{\partial v^i} + s_{\addspace{k}ij}^k \frac{\partial L}{\partial v^k} v^j - \Lie{X_i}{L} 
\label{eq:hamel}
\end{equation}
where $s_{ij}^k:\Q \rightarrow \R$ are the structure constants satisfying:
\[
s^k_{\addspace{k}ij} X_k = [X_i, X_j]
\]
with $[X_i,X_j]$ the Lie bracket between vector fields. In the absence of nonconservative forces $\tau$, any solution $q(t)$ to \eqref{eq:hamel} will be an extremal curve of the action integral
\begin{equation}
\textstyle{ \int_{t_0}^{t_f} L(q(t) , v(t) ) {\rm d}t} 
\label{eq:action}
\end{equation}
under all variations with endpoints fixed.
\end{remark}

In the common case when $X_1, \ldots, X_n$ are chosen to be coordinate vector fields, generalized coordinates $q^1, \ldots, q^n$ can be adopted. In this case, $[X_i, X_j]=0$ pairwise (i.e., all structure constants are zero), and so \eqref{eq:hamel} simplifies to the Euler-Lagrange equations. 
While there are many ways of defining the Coriolis matrix for these equations, the particular choice \eqref{eq:christoffel_C} using Christoffel symbols can have benefits, as discussed in the next subsection. Unfortunately, calculating the Christoffel symbols symbolically is not tractable for high-degree-of-freedom systems. 

The symbols can be computed numerically for some restricted classes of mechanisms (e.g., those with prismatic and revolute joints) \cite{echeandia2021numerical}, but other systems with closed kinematic loops are not treated by available algorithms. Further, humanoid and quadruped robots are frequently modeled with a 6-DoF free joint between the world and their torso that is modeled using generalized velocities to avoid representation singularities associated with Euler angles. 
We aim to accommodate these general settings, and do so by making use of a geometric treatment of the problem.

\subsection{Riemannian Geometry Through the Lens of Mechanics}
\label{sec:riemannian_geometry}

The kinetic energy $T(q,v)= \frac{1}{2} v\T H(q) v$ allows us to associate each point $q$ with an inner product on $T_q\Q$. We call this collection of inner products a metric, denoted $\mathbb{G}$, with the maps written as $\l \cdot, \cdot \r_q : T_q\Q \times T_q\Q \rightarrow \mathbb{R}$. We make this correspondence such that $\l \dot{q}, \dot{q} \r_q = v\T H(q) v$ for any $\dot{q} = X_i(q) v^i
 \in T_q \Q$.

\begin{figure}
    \centering
    \includegraphics[width=.6\columnwidth]{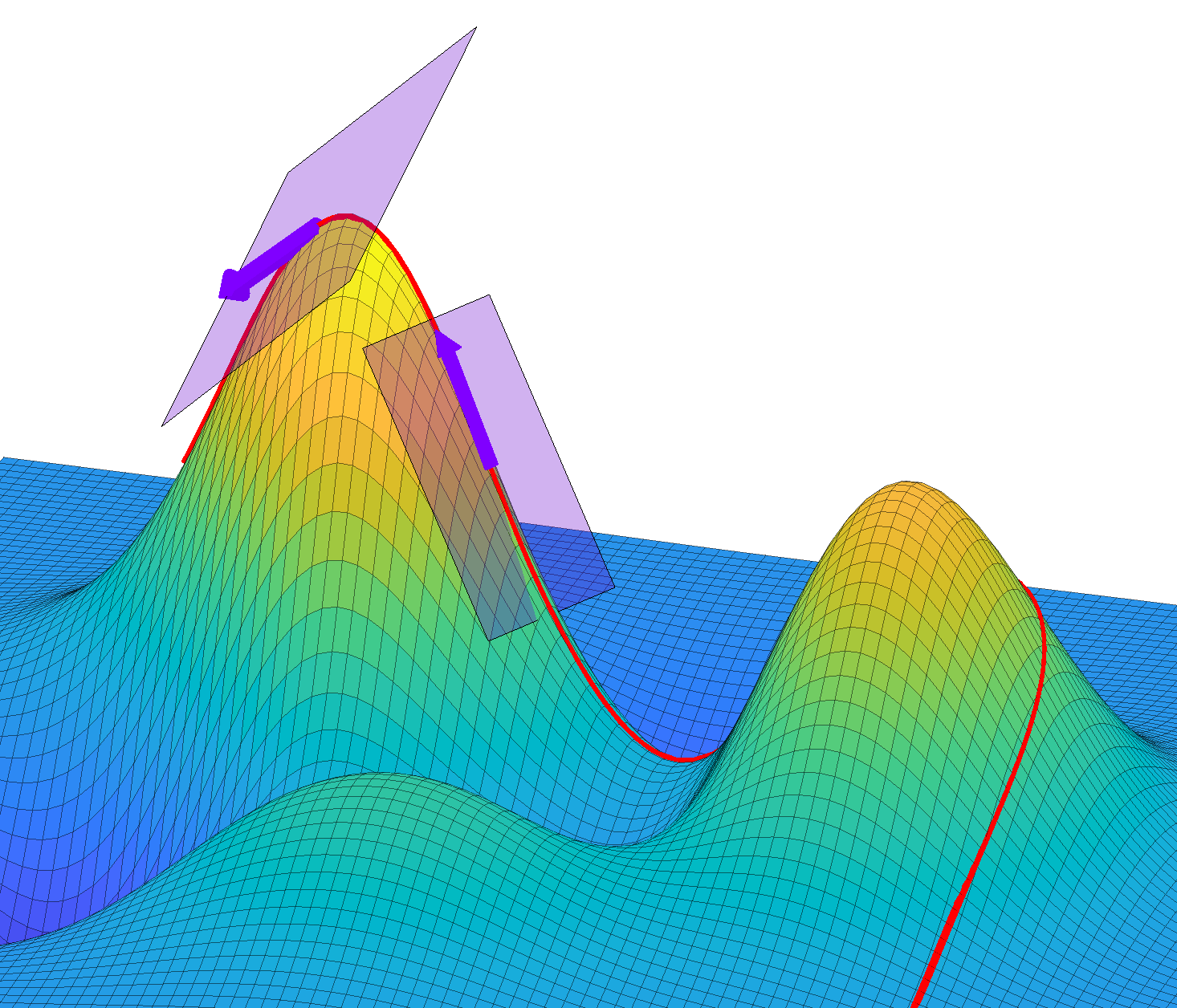}
    \caption{Motion of a particle on a frictionless surface in the absence of gravity. Equivalently, a geodesic on the surface. The particle acceleration is always normal to the surface (since constraint forces are normal to it). Equivalently, the covariant derivative of its velocity vector w.r.t. the Riemannian connection is zero along the curve.}
    \label{fig:surf}
\end{figure}

In the absence of potential forces, the action \eqref{eq:action} simplifies, and trajectories of the system must trace extremal curves of:
\[
\textstyle{\int_{t_0}^{t_f}} \l \dot{q}(t), \dot{q}(t) \r_{q(t)} {\rm d} t
\]
under variations with endpoints fixed. Such curves correspond to geodesics -- paths that minimize distances according to the metric $\mathbb{G}$. 
Returning to our mechanics interpretation, let us consider the physics of such an extremal curve via an example of a point mass evolving on a frictionless 2D manifold embedded in $\mathbb{R}^3$ (as drawn in Fig.~\ref{fig:surf}). With no friction or gravity, the only forces are from constraints, which act normal to the surface. 
Consequently, the particle's acceleration is always normal to the surface, and its speed, proportional to $\sqrt{\l \dot{q}, \dot{q} \r}_q$, remains constant along the geodesic.


Put another way, for a point mass moving under the influence of constraint forces alone, the projection of the acceleration onto the local tangent plane is always zero. This insight motivates us to define what we call a {\em covariant derivative} for the velocity vector along the curve (a directional derivative of sorts), which is obtained by projecting the acceleration vector at each point onto the local tangent plane. For the curve to be a geodesic, the covariant derivative of its velocity must be zero along the curve. 

More generally, we consider our manifold $\Q$ with the metric $\G$ (i.e., such that $(\Q,\G)$ defines a Riemannian manifold). We further consider $\Q$ isometrically embedded in some Euclidean space, where the length of any path according to the metric is simply the usual Euclidean length in the ambient space\footnote{Such embeddings are always possible according to the Nash embedding theorem \cite{nash1956imbedding}.}. Using a similar construction as above, we can give the rate of change of a vector field $Y$ in the direction of another $X$ as $\nabla_X^\star Y$ via taking a conventional directional derivative of $Y$ in the ambient space, and projecting the result, pointwise, to the local tangent space. This manner of differentiating vectors via projection to {\em connect} nearby tangent spaces is natural, and for this reason, we call $\nabla^\star$ \underline{the} Riemannian connection associated with $(\Q,\G)$.  For any pair of smooth vector fields $X, Y\in\X(\Q)$ and any scalar-valued function $f$ on $\Q$ this operation satisfies; 
\begin{enumerate}[label=A\arabic*), ref=A\arabic*)]
    \item $\nabla^\star_{fX}Y = f \nabla^\star_{X} Y$ 
    \item $\nabla^\star_X f Y = f \nabla^
    \star_X Y + ( \Lie{X}{f} ) Y$ 
\end{enumerate}

Any bilinear mapping $\nabla: \X(\Q) \times \X(\Q) \rightarrow \X(\Q)$ (i.e., not just the Riemannian connection) satisfying the above two properties is called an affine connection, and defines an alternative manner for taking directional derivatives of a vector field. 
For a vector field $Z$ defined only along a curve $\gamma: \R \rightarrow \Q$ with $Z(s) \in T_{\gamma(s)} \Q$, the covariant derivative of $Z$ along $\gamma$ is denoted $\nabla_{ \dot{\gamma} } Z$. A curve $\gamma$ is a geodesic of the connection if $\nabla_{\dot{\gamma}} \dot{\gamma} = 0$. 
Notably, geodesics of the Riemannian connection $\nabla^\star$ enjoy a special property in that they correspond with geodesics curves (i.e., of extremal length) for the metric $\mathbb{G}$. In general, geodesics of an arbitrary connection need not be geodesic curves of the metric.

An affine connection also defines a manner of parallel transporting tangent vectors along a curve. For a vector field $Z$ along a curve $\gamma$, we say that the collection of tangent vectors $Z(s) \in T_{\gamma(s)} \Q$ are parallel according to the connection if $\nabla_{\dot{\gamma}}Z =0$. Given some initial $Z(0)$, the condition $\nabla_{\dot{\gamma}}Z =0$ gives a differential equation for $Z(s)$ that then uniquely determines how to parallel transport $Z(0)$ along the curve. 
Note that parallel transport for the Riemannian connection relies on local projections to connect nearby tangent spaces. Since these projections are distance optimizing, parallel transport according to $\nabla^\star$ can be viewed as preventing extra spinning of $Z$ around $\gamma$, and $\nabla^\star$ is said to be torsion-free.

It is worth mentioning that while we have introduced operations of covariant differentiation and parallel transport using an extrinsic view of a manifold $\Q$ (i.e., by considering an embedding of it in Euclidean space), these operations do not depend on the details of that embedding, and so can be viewed as ones defined intrinsically on the manifold itself.

To this end, we can represent $\nabla^\star$ via a set of components at each point using our basis vector fields $X_1, \ldots, X_n$. Let us consider a choice of coordinate vector fields with symbols $\Gamma_{ijk}^\star $ defined by \eqref{eq:christoffel}. We can relate these symbols to those of the second kind via 
\[
\Gamma^{i \star}_{\addspace{i}jk} = H^{i \ell} \Gamma_{\ell jk}^\star
\]
where $H^{ij}$ denotes $( H^{-1})_{ij}$. These symbols can be shown to give the covariant derivative of any basis vector field along another, according to:
\[
\nabla_{X_j}^\star X_k = \Gamma^{i \star}_{\addspace{i}jk} X_i 
\]
These relationships then uniquely specify the covariant derivative w.r.t. the Riemannian connection on any generic pair of vector fields via properties A1 and A2.

\subsection{Properties of Affine Connections}
\iftransportfig
\begin{figure*}
\center \includegraphics[width=.7 \textwidth]{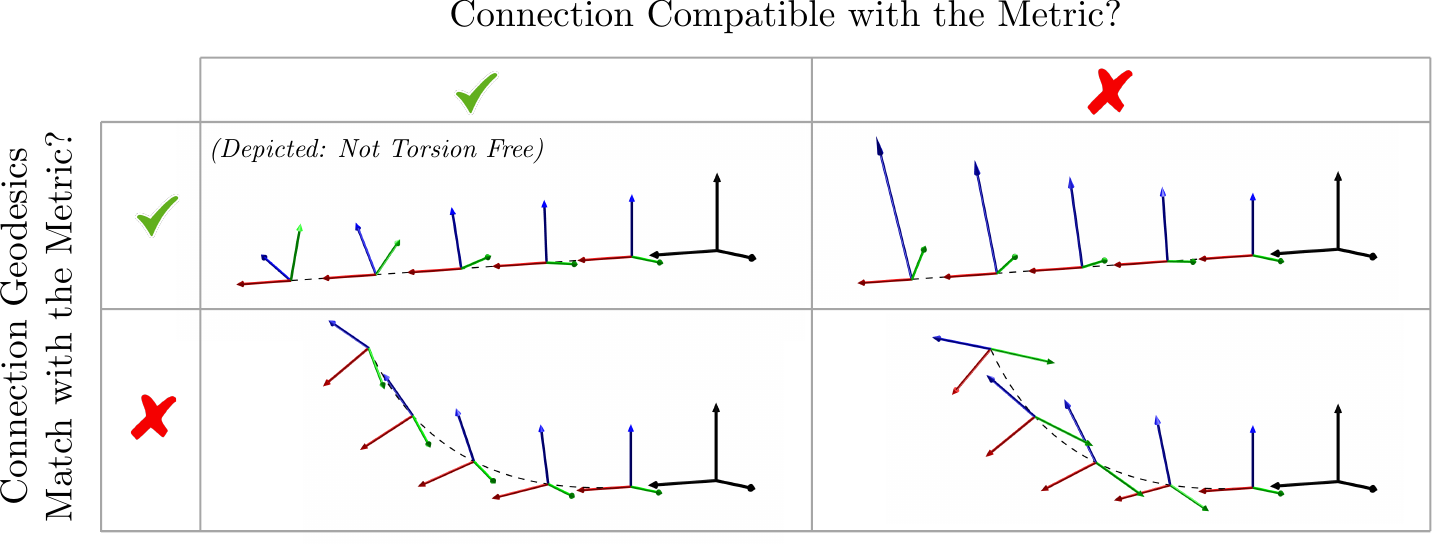}
\caption{We consider parallel transport along the geodesics of four different affine connections. In each case, three vectors are transported -- they start mutually orthogonal, and the geodesic is considered along the initial red vector. The connection is compatible with the Euclidean metric (left column) if the coordinate system remains orthonormal throughout. This notion of being compatible with the metric is distinct from the geodesics of the connection being the same as those of the Euclidean metric (top row). The upper left corner shows a connection with torsion. In this case, the unique torsion-free metric-compatible connection  would instead keep the coordinate system at a constant orientation.}
\label{fig:transport}
\end{figure*}
\fi
More generally, we can uniquely specify an affine connection $\nabla$ by fixing its connection coefficients\footnote{We reserve the name "Christoffel Symbols" to refer to the coefficients of the Riemannian connection.} $\Gamma_{ijk} :\Q \rightarrow \R$ and defining the covariant derivative w.r.t. $\nabla$ according to $\nabla_{X_j} X_k = \Gamma^i_{jk} X_i$. Given the immense flexibility this provides, we consider a few desirable properties 
\iftransportfig
(illustrated in Fig.~\ref{fig:transport}) 
\fi
for a connection, and whose relationships are diagrammed in Fig.~\ref{fig:summary}. 

\subsubsection{Metric Compatibility} A connection is said to be compatible with the metric if $\forall X,Y,Z \in \X(\Q)$
\[
\Lie{X}{ \l Y , Z \r } = \l \nabla_X Y, Z \r + \l Y, \nabla_X Z \r\,,
\]
which ensures, in particular, that any orthonormal set of tangent vectors (rooted at the same base point) will remain orthonormal under parallel transport. 
As would be expected, it can be shown that the Riemannian connection $\nabla^\star$ is compatible with its metric $\G$. Given its canonical role, we consider the difference between the Riemannian connection and another connection according to:
\[
D(X,Y,Z) = \l X, \nabla^\star_Y Z \r - \l X, \nabla_Y Z \r
\]
which we call a contorsion tensor \cite{nakahara2003geometry}, with components:
\[
D_{ijk}= D(X_i, X_j, X_k) = \Gamma_{ijk}^\star -\Gamma_{ijk}
\]

\begin{proposition}[Metric-Compatible Connections]
A connection  $\nabla$ is compatible with a metric $\G$ iff its associated contorsion tensor $D$ is anti-symmetric in the first and third arguments (i.e., $D(X,Y,Z) = - D(Z,Y,X)$). Equivalently, the components $D_{ijk}$ are anti-symmetric in indices (1,3).
\label{prop:metric_compat_connections}
\end{proposition}

\begin{proof}
See \cite[Section 7.2.6]{nakahara2003geometry} for the setup of an exercise on this result.
\end{proof}

A connection is said to be torsion-free if $\forall X, Y \in \X(\Q)$
\[
[X,Y] = \nabla_X Y - \nabla_Y X
\]
As motivated previously, a connection being torsion-free roughly ensures that vectors do not twist around a curve when parallel transported. This condition is the same as requiring that the structure constants $s^i_{jk}$ satisfy 
$s^i_{jk} = \Gamma^i_{jk} - \Gamma^i_{kj}
$, 
which reduces to $\Gamma^{i}_{jk} = \Gamma^i_{kj}$ or equivalently $\Gamma_{ijk} = \Gamma_{ikj}$ when working with coordinate vector fields. Considerations of connections being metric-compatible and torsion-free lead to the fundamental theorem of Riemannian geometry:


\begin{theorem}[Riemannian Connection, \cite{do1992riemannian}] On a Riemannian manifold $\Q$ there exists a unique affine connection $\nabla$ that is compatible with the metric and torsion free. This connection is the Riemannian (or Levi-Civita) connection $\nabla^\star$.
\end{theorem}

\begin{remark}
See \cite[App.~1.D]{arnol2013mathematical} for an alternative description of metric-compatible torsion-free parallel transport. 
\end{remark}

With this result, the Riemannian connection can be uniquely specified by the Koszul formula, which gives that for any smooth vector fields $X,Y,Z$ \cite[Thm.~3.6]{do1992riemannian}
\begin{align*}
2 \l X , \nabla^\star_Y Z \r = &  \Lie{Z}{ \l X, Y \r } + \l Z, [X,Y] \r \\
                         & + \Lie{Y}{ \l X, Z \r } + \l Y, [X, Z] \r \\
                         & - \Lie{X}{ \l Y, Z \r } + \l X, [Y, Z] \r 
\end{align*}
This definition then allows us to consider generalized Christoffel symbols \cite{bullo2004geometric} when working with non-coordinate basis vector fields, which are given according to:
\newcommand{\myscale}{.97}
\begin{align}
\scaleobj{\myscale}{2\Gamma_{ijk}^\star }&\scaleobj{\myscale}{= 2 \l X_i, \nabla_{X_j}^\star X_k \r} \label{eq:gen_chris}\\
&\scaleobj{\myscale}{=\Lie{X_k}{H_{ij}} + \Lie{X_j}{ H_{ik} }  - \Lie{X_i}{H_{jk}} } \nonumber \\
                             &\scaleobj{\myscale}{~~ +\l X_k, [X_i,X_j] \r + \l X_j, [X_i,X_k] \r + \l X_i, [X_j,X_k] \r } \nonumber \\
&\scaleobj{\myscale}{=\Lie{X_k}{H_{ij}} + \Lie{X_j}{ H_{ik} }  - \Lie{X_i}{H_{jk}} + s_{kij} + s_{jik} + s_{ijk} } \nonumber
\end{align}
where $s_{ijk} = H_{i \ell} s^\ell_{jk}$. If you removed the last term $s_{ijk}$ on the last line of \eqref{eq:gen_chris}, what remains would be symmetric in its (2,3) indices. So, letting $\Gamma^\star_{i (jk)} = \frac{1}{2} ( \Gamma^\star_{ijk} + \Gamma^\star_{ikj})$ denote the symmetrized version, we have
\begin{equation}
\Gamma^\star_{i jk} = \Gamma^\star_{i (jk)} + \frac{1}{2} s_{ijk}
\label{eq:gamma_symmetry}
\end{equation}
We'll see the significance of this result in Sec.~\ref{sec:diff}.

\subsubsection{Geodesic Agreement} We note that a connection being ``compatible with the metric" is different from requiring the geodesics of the connection to match those of the metric. We further build on this distinction through the following results.

\begin{proposition}[Geodesic Agreement]
A connection gives the same geodesics as the Riemannian connection iff its contorsion tensor is anti-symmetric in its second and third arguments (i.e., $D(X,Y,Z) = -D(X,Z,Y)$. Equivalently, the components $D_{ijk}$ are anti-symmetric in indices (2,3).
\end{proposition}

Out of all possible connections with Riemannian geodesics, the Riemannian connection is the only one that is torsion-free. 

\begin{theorem}
Given any affine connection on $\Q$, there is a unique torsion-free connection with the same geodesics. 
\label{thm:geodesic_freedom}
\end{theorem}
\begin{proof}
See \cite[Chapter 6, Corollary 17]{spivak2005comprehensive}.
\end{proof}

If we combine Prop.~\ref{prop:metric_compat_connections} and Thm.~\ref{thm:geodesic_freedom}, we see that we still have a great deal of freedom even when restricting to connections that a) are compatible with the metric and b) give the same geodesics as the metric.

\begin{corollary} 
\label{cor:antisymmetric}    
A connection is metric-compatible and gives the same geodesics as the Riemannian connection iff its contorsion tensor $D(X,Y,Z)$ is totally anti-symmetric (i.e., in any pair of arguments). Equivalently, the components $D_{ijk}$ are anti-symmetric in any pair of indices.  
\end{corollary} 

\begin{corollary}
\label{corr:InfinityOfConnections}
For manifolds of dimension one or two, there is a unique connection (namely the Riemannian connection) that is metric-compatible and gives the same geodesics as the Riemannian connection. In dimensions three or higher, there are infinitely many such connections.
\end{corollary}

\begin{proof}
\ifshortproof
The space of totally antisymmetric $(0,3)$ tensor fields on $\Q$ \cite{bullo2004geometric} is isomorphic to the space of smooth 3-forms on the manifold. There are no three forms on a manifold when $n=1,2$, but an infinite number exist when $n\ge3$.
\fi
\end{proof}

Referring back to Fig.~\ref{fig:summary}, we see that the upper-left of the diagram displays a large area of overlap between those connections that are metric compatible and those whose geodesics are Riemannian, with the Riemannian connection residing in this intersection. 


\section{Relationship between Coriolis Factorizations and Affine Connections}
\label{sec:link}

This section creates a bridge between the properties of an affine connection and those of a Coriolis factorization that we derive from it. 
We first note that when working with coordinate vector fields, the Christoffel symbols are symmetric in the (2,3) indices, which means that we have the equivalence $C_{ij}^\star = \Gamma_{ijk}^\star v^k = \Gamma_{ikj}^\star v^k$. When working with generalized velocities, we can lose the $(2,3)$ symmetry of the symbols, and instead need to choose: Do we define our Coriolis matrix associated with the Riemannian connection via  $C_{ij}^\star =  \Gamma^\star_{ijk} v^k$, $\Gamma^\star_{ikj} v^k$, or the symmetric part of the symbols $\Gamma^\star_{i(jk)} v^j$? We find that, for general connections, the fruitful pairing is to choose:
\begin{equation}
C_{ik}(q,v) = \Gamma_{ijk}(q)\, v^j = \l X_i, \nabla_{\dot{q}} X_k \r
\label{eq:c_def}
\end{equation}
in the sense that this definition allows us to directly relate properties of the Coriolis matrix with those of its connection. 

Our first result concerns passivity and metric compatibility.

\begin{proposition}[Skew-symmetry and Metric Compatibility] 
\label{prop:MatrixCompatibility}
Consider a Coriolis matrix definition $C(q,v)$ and a connection $\nabla$ related via \eqref{eq:c_def}. The connection is compatible with the metric iff $\dot{H}-2C$ is skew-symmetric pointwise. 
\end{proposition}

\begin{proof} We consider the forward implication via computation:
    \begin{align}
    &\phantom{=} \dot{H}_{ij} - C_{ij}-C_{ji} \label{eq:forward_passive} \\ & = \Lie{X_k}{ \l X_i, X_j \r } \, v^k - \l X_i ,\nabla_{X_k} X_j \r v^k - \l X_j , \nabla_{X_k} X_i \r v^k \nonumber \\
    &= \left ( \l \nabla_{X_k} X_i, X_j \r + \l X_i, \nabla_{X_k} X_j\r \right) v^k \nonumber \\
    & ~~~~~~ - \left(\l X_i ,\nabla_{X_k} X_j \r  + \l X_j , \nabla_{X_k} X_i \r \nonumber \right) v^k =0
\end{align}
where we have used the metric compatibility when going from the second to the third line.

For the reverse direction,  the second line of \eqref{eq:forward_passive} is assumed zero, which gives that $\Lie{X_i}\l X_j, X_k \r = \l \nabla_{X_i} X_j , X_k \r + \l X_j, \nabla_{X_i} X_k \r$. To show metric compatibility, we must show that the same type of property holds on arbitrary vector fields. We consider the triple of vector fields $X,Y,Z \in \X(\Q)$ expressed as $X=X_i x^i$, $Y=X_i y^i$, $Z=X_i z^i$. Via using properties of the Lie derivative, we have
\renewcommand{\myscale}{.96}
\begin{align*}
\scaleobj{\myscale}{\Lie{X}{ \l Y , Z \r}} & \scaleobj{\myscale}{= \Lie{X_i x^i}{\l X_j y^j, X_k z^k \r}} 
\\ &\scaleobj{\myscale}{= x^i ( \Lie{X_i}{y^j z^k} ) \l X_j, X_k \r + x^i y^j z^k \Lie{X_i}{\l X_j , X_k \r}}
\end{align*}
We use the second line of \eqref{eq:forward_passive} to carry out the Lie derivative $\Lie{X_i}\l X_j, X_k \r$, which ultimately leads to:
\begin{align*}
\nabla_X \l Y , Z \r &= x^i z^k \l \nabla_{X_i} X_j y^j, X_k \r + x^i y^j \l X_j, \nabla_{X_i} X_k z^k \r \\
&= \l \nabla_X Y , Z \r + \l Y, \nabla_X Z \r
\end{align*}
This proves the reverse implication.
\end{proof}    

\begin{remark}
The forward implication was discussed in \cite[pg. 94]{van2017} when working with coordinate vector fields. 
\end{remark}

Our second result considers whether the definition of the Coriolis matrix gives the correct equations of motion:
    
\begin{proposition} 
\label{prop:MatrixGeodesics}
Consider a Coriolis matrix definition $C(q,v)$ and a connection $\nabla$ related by \eqref{eq:c_def}. The connection gives the same geodesics as the Riemannian connection iff the Coriolis matrix satisfies $C v = C^\star v$, where $C^\star$ is the factorization corresponding to the Riemannian connection.
\end{proposition}

A connection giving the same geodesics as the Riemmannian connection is equivalent to the associated factorization giving the correct dynamics. This follows since reoganizing Hamel's equations \eqref{eq:hamel} via the Koszul formula \eqref{eq:gen_chris} provides:
\[
\frac{{\rm d}}{{\rm d}t} \frac{\partial L}{\partial v^i} + s_{\addspace{k}ij}^k \frac{\partial L}{\partial v^k} v^j = H_{ij} \dot{v}^j + \Gamma^\star_{ijk} v^j v^k = (H \dot{v} + C^\star v)_i
\]

\begin{corollary} Any Coriolis factorization that satisfies the skew-symmetry property and gives the correct dynamics is derived from a connection whose coefficients differ from the Riemannian connection by a totally anti-symmetric contorsion tensor. There is a unique such factorization when $n=1,2$, but an infinite number of such factorizations when $n\ge3$.
\end{corollary}

\begin{proof}
The result follows from combining Corollary \ref{corr:InfinityOfConnections} with Propositions \ref{prop:MatrixCompatibility} and \ref{prop:MatrixGeodesics}.
\end{proof}

\begin{remark}
Coriolis factorizations also play a role in describing the evolution of the generalized momentum within the Hamiltonian formalism. When adopting generalized coordinates, setting $p_i = \frac{\partial L}{\partial \dot{q}^i}$ we have:
\[
\dot{p}_i = \frac{\partial L}{\partial q^i} + \tau_i
\] 
More generally, however, when adopting generalized speeds and assigning $p = H(q) v$, these equations take the form:
\[
\dot{p} = \tau + C(q,v)\T v - g(q)
\]
where $C(q,v)$ is any Coriolis factorization satisfying the skew symmetry property \eqref{eq:skew}.
\end{remark}

\subsection{Interpretation in Passivity-Based Control}
\subsubsection{Setup} 

In this subsection, we consider the use of these Coriolis matrices in passivity-based control, and explore the role that the torsion-free property has on the resulting trajectories. To do so, we employ the tracking control component of \cite{slotine1987adaptive}, which applies to the case of working with generalized coordinates. 
Given a desired trajectory $q_d(t)$, we consider the error with coordinates $e^i = q^i - q_{d}^i$ and corresponding vector field $e = X_i e^i$. We then consider a time-varying vector field $\dot{q}_r(q, t) = X_i(q) \,\dot{q}_d^i(t) - \Lambda \, e(q, t)$ where $\Lambda>0$ is a scalar gain that is used to control convergence speed. As before, we use $v\in \mathbb{R}^n$ to represent the components of $\dot{q} \in T_q \Q$, and likewise we use $v_r \in \mathbb{R}^n$ to represent the components of $\dot{q}_r \in T_q Q$.


Following \cite{slotine1987adaptive}, we form the sliding variable:
\[
s = v-v_r
\]
where the task of tracking the desired trajectory $q_d(t)$ can now be achieved via regulating $s\rightarrow 0$.

This definition leads to a certainty-equivalence control law
\begin{equation}
\tau = \hat{H}(q) \dot{v}_r + \hat{C}(q,v) v_r + \hat{g}(q) - K_D s
\label{eq:certainty_equiv_law}
\end{equation}
where $\hat{H}$, $\hat{C}$, and $\hat{g}$ represent estimates of the mass matrix, Coriolis matrix, and gravity terms, and $K_D \succ 0$ is a positive-definite user-chosen gain matrix. We consider estimates for the classical inertial parameters of the mechanism \cite{atkeson1986estimation} as $\hat{\theta} \in \mathbb{R}^{10 N_B}$ for a system of $N_B$ rigid bodies. This allows expressing the control law \eqref{eq:certainty_equiv_law} above as:
\begin{equation}
\label{eq:certainty}
\tau = Y(q,v, v_r, \dot{v}_r) \hat{\theta} - K_D s
\end{equation}
where $Y(q,v, v_r, \dot{v}_r) \in \mathbb{R}^{n \times (10 N_B)}$ is often referred to as the Slotine-Li regressor matrix \cite{yuan1995recursive,garofalo2013closed}. 

Considering the closed-loop behavior under this control law, the sliding variable then evolves as \cite{slotine1987adaptive}
\[
H(q) \dot{s} + (C(q,\dot{q})+K_D) s = Y(q,v, v_r, \dot{v}_r) \tilde{\theta}
\]
where $\tilde{\theta} = \hat{\theta}-\theta$ gives the parameter error. To study the regulation of $s\rightarrow 0$, we consider the Lyapunov candidate
\[
V = \frac{1}{2} s\T H(q) s + \frac{1}{2} \tilde{\theta}\T \tilde{\theta}
\]
which, when using {\em any} factorization satisfying $\dot{H} = C+C\T$, can be shown to have rate of change:
\begin{equation}
\dot{V} = - s\T K_D s + s\T Y \tilde{\theta} + \tilde{\theta}\T \dot{\tilde{\theta}} \label{eq:Vdot}
\end{equation}
When parameter estimates are fixed and the parameter error $\tilde{\theta} \equiv 0$, we have convergence of $s\rightarrow 0$ as $t \rightarrow \infty$. When $\tilde{\theta}$ is non-zero, we can adapt the estimate \cite{slotine1987adaptive} according to
\begin{equation}
\dot{\hat{\theta}} = - Y^T s\,,
\end{equation}
which again recovers $\dot{V} = - s\T K_D s$. For the purposes of the following example, we ignore parameter adaptation and leave a persistent parameter error to see how the controller performs with modeling inaccuracy.


\begin{remark}
The above construction created the reference velocity field $\dot{q}_r(q,t)$ to provide convergence to $q_d(t)$. However, one can consider generating this reference velocity field in many other ways. For example, it could be constructed to provide convergence to a submanifold, convergence in task-space \cite{niemeyer1991performance,slotine1991general}, or to a desired limit cycle, with $\dot{v}_r \in \mathbb{R}^n$ then defined based on the time differentiation of the coordinates of this vector field along the trajectory.

Alternatively, as discussed in \cite[Example 3.4]{slotine2003modular} the reference dynamics can be specified at the acceleration level via defining desired accelerations $\dot{v}_d(t,q,v)$ directly and then constructing $v_r$ via integrating the dynamics:
\[
\dot{v}_r(t) = -k (v_r(t) - v(t)) + \dot{v}_d(t, q(t), v(t))\,.
\]
In that case, it remains that $s\rightarrow 0$ under \eqref{eq:certainty}. 

\end{remark}


\subsubsection{Example} As motivated in the introduction, we consider the simple example of tracking a trajectory for a point mass in 3D. We estimate the point mass as $\hat{m} = 0.9$ kg with the true mass $m=1$ kg. Generalized coordinates $q=[p_x, p_y,p_z]\T$ are adopted for the positions of the mass in a Euclidean coordinate system. The tracking task is chosen as following a line with varying velocity such that $q_d(t) = [t + \sin(t),~0,~0]\T$. 

In the first case, we use the Coriolis matrix $C=0$ which corresponds to the torsion-free connection, and we use parameters $\Lambda=1$ and $K_D=1$. The resulting trajectories are shown vs. time in Fig.~\ref{fig:state_traj_point_mass}. The results are uneventful, and the trajectory converges to near-zero error. The error does not exactly approach zero due to model mismatch. 

As a second case, we consider the Coriolis factorization $C=-5 (v\times)$, where $(v \times) \in \mathbb{R}^{3\times 3}$ is the skew-symmetric cross-product matrix so that for any other $w\in \mathbb{R}^3$, $(v \times) w = v\times w$. This factorization may feel unreasonable since the Coriolis terms are zero in this case. However, since $H(q) = m \mathbf{1}_{3}$ where $\mathbf{1}_{3}$ the identity matrix, we have $\dot{H} = 0$, and so any skew-symmetric $C$ will satisfy the skew property \eqref{eq:skew}. Using any scalar multiple of the Cartesian cross-product matrix $(v\times)$ (or more generally any scalar configuration-dependent function times the cross-product matrix) therefore satisfies \eqref{eq:factorization} and \eqref{eq:skew}. The resulting trajectory is plotted in Fig.~\ref{fig:state_traj_point_mass_torsion}. The trajectory has higher frequency content than the previous case and thus would be comparatively undesirable for deployment on a physical system. More extreme factorizations (e.g., $-10 (v\times)$) yield trajectories with even higher frequency content.

\begin{figure}
    \centering
    \includegraphics[width=\columnwidth]{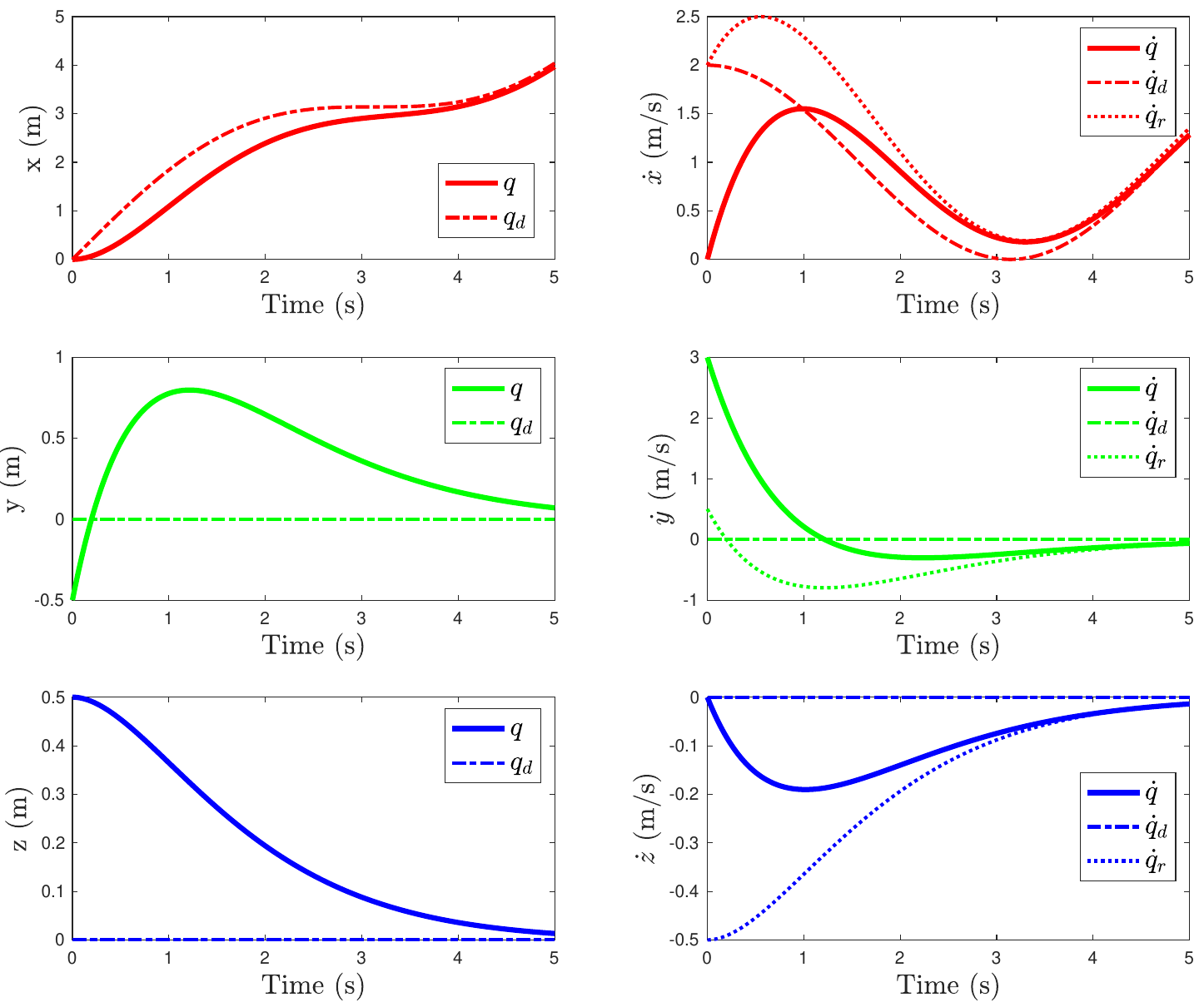}
    \caption{Point mass state trajectories, torsion-free. \later{bigger font}}
    \label{fig:state_traj_point_mass}
\end{figure}

\begin{figure}
    \centering
    \includegraphics[width=\columnwidth]{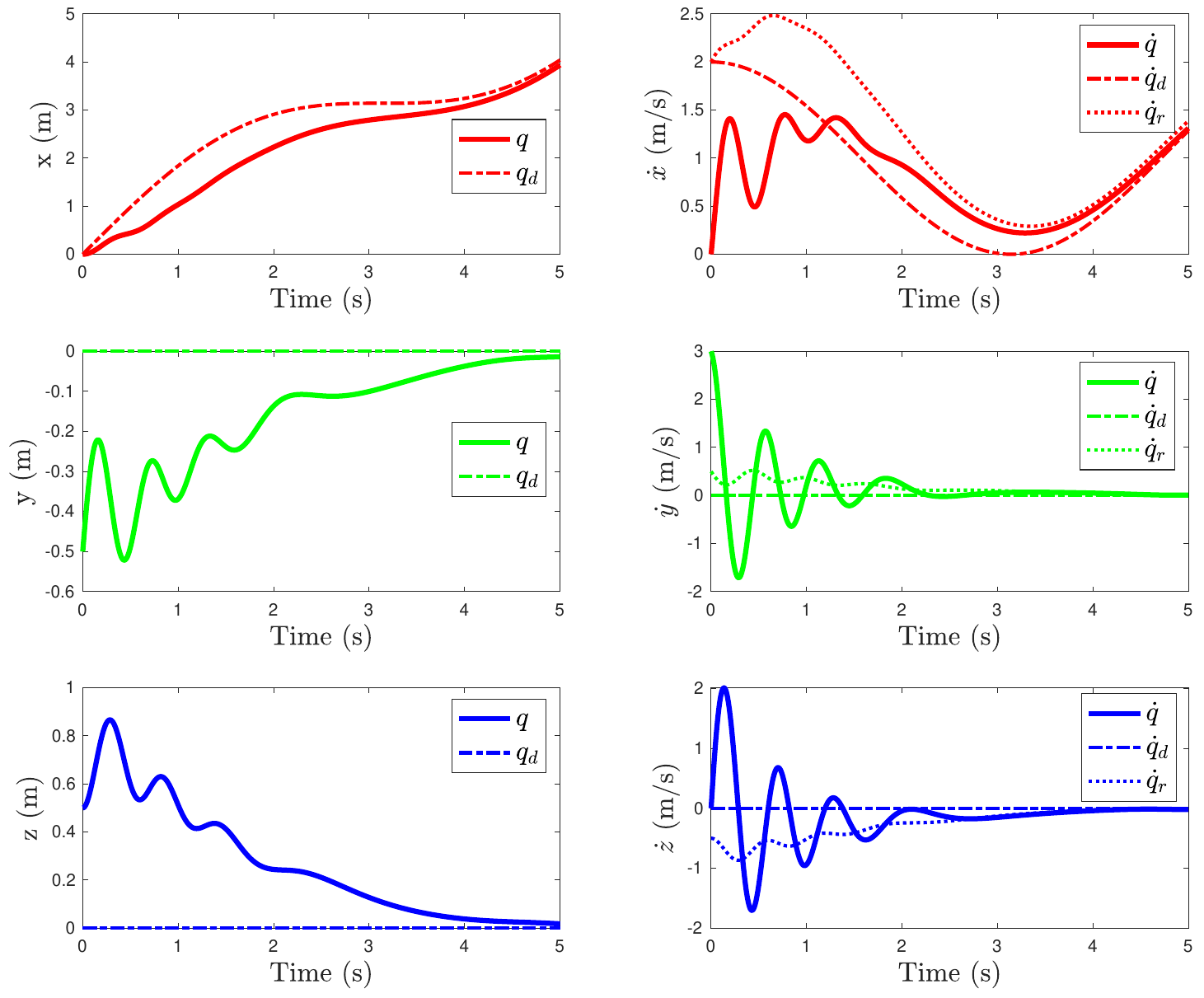}
    \caption{Point mass state trajectories, non-torsion-free. \later{bigger font}}
    \label{fig:state_traj_point_mass_torsion}
\end{figure}

The effects of torsion are better demonstrated in Fig.~\ref{fig:parametric_point_mass}, which shows a parametric trajectory in 3D space. The point mass starts nearest (at $x=0$) and moves away (up and to the right in the figure) as time progresses. At the beginning of the trajectory, the mass has a velocity to the left, while the reference velocity is moving forward. The terms from the Coriolis factorization $-5 (v\times) v_r$ cause a counterproductive upward acceleration, sending the mass further from the target. \later{There is some benefit to the twisting of the trajectory in that the velocity error is rapidly changing direction, and so does not get as much time to accumulate as position error. Indeed, the RMS position error for the case with torsion (XXX) is X\% lower than the torsion-free case (XXX).} 

\begin{figure}
    \centering
    \includegraphics[width=\columnwidth]{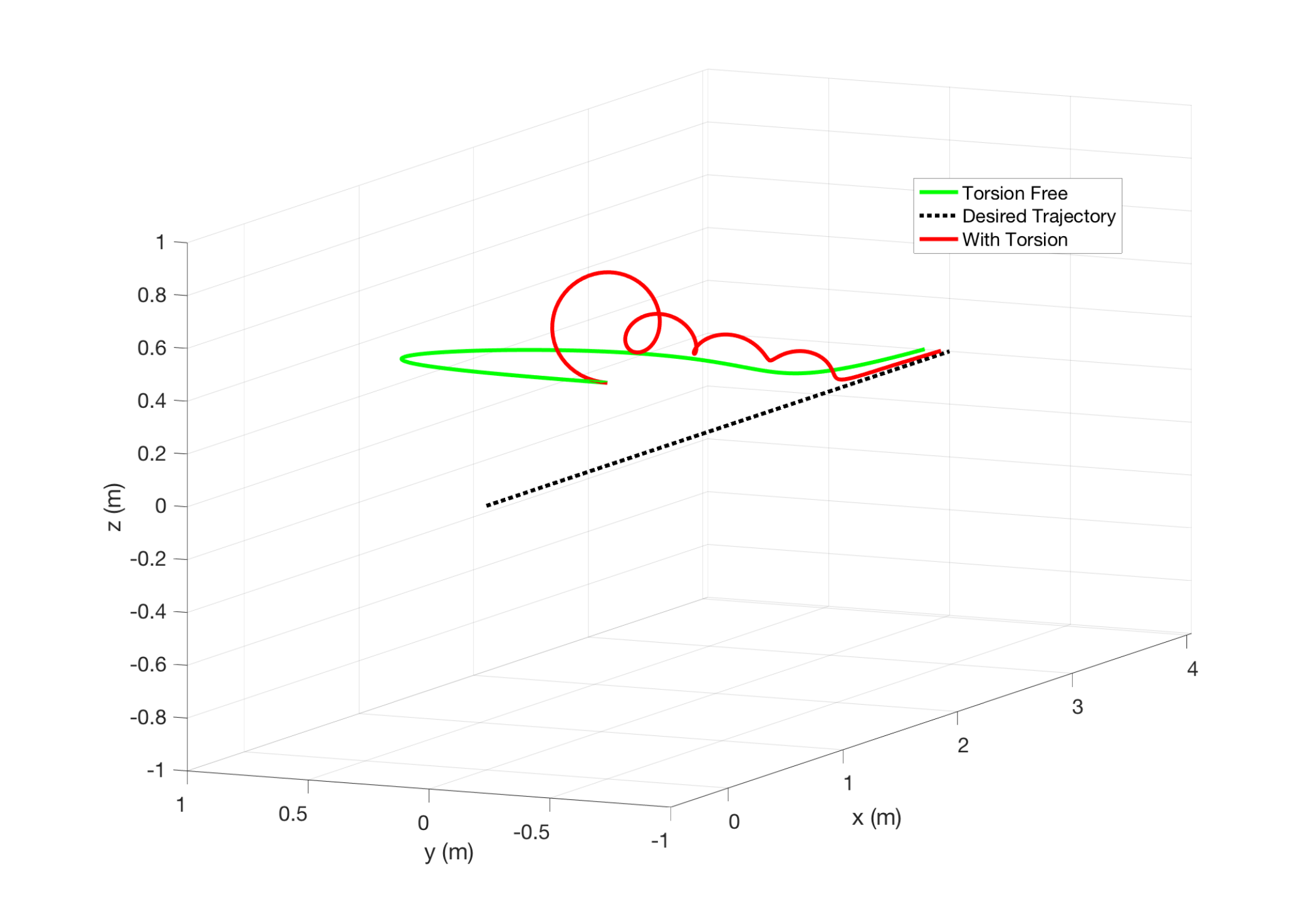}
    \caption{Point mass trajectories, torsion free vs. with torsion. \later{bigger font}}
    \label{fig:parametric_point_mass}
\end{figure}

\begin{figure}
    \centering
    \includegraphics[width=.75 \columnwidth]{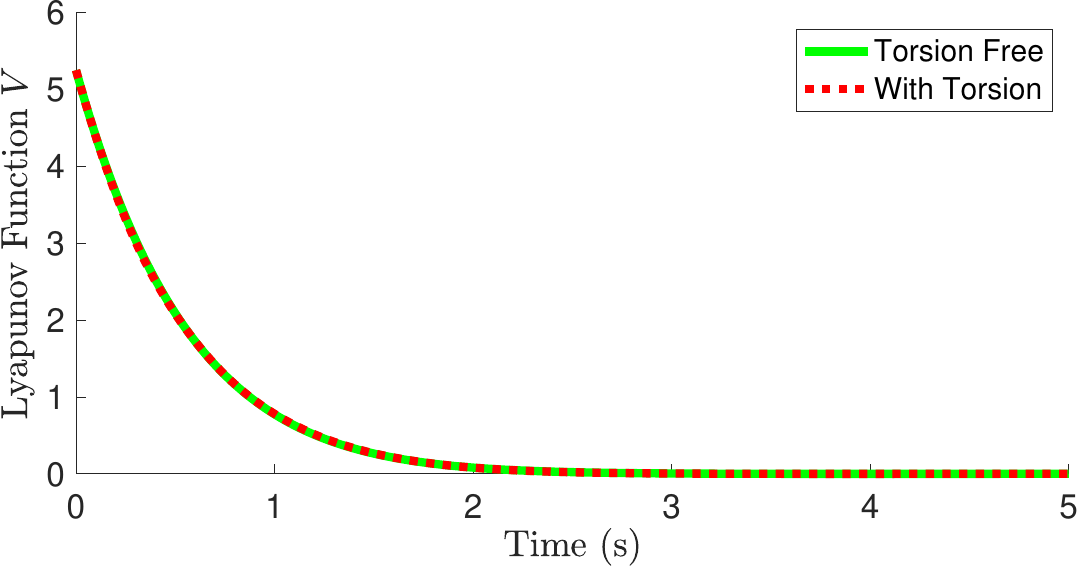}
    \caption{Lyapunov function evolution for point-mass example: torsion-free vs. with torsion. \later{bigger font}}
    \label{fig:Lyap_point_mass}
\end{figure}

This difference in behavior between the two cases is masked by conventional Lyapunov analysis, as \eqref{eq:Vdot} is independent of the choice of $C$. Fig.~\ref{fig:Lyap_point_mass} shows the Lyapunov function evolution for both setups, and in this case, the evolution is nearly identical (the evolution is not exactly identical, which becomes more apparent for larger mismatch in the mass estimate). While this finding may be surprising, it is worth noting that a Lyapunov function value provides a condensed scalar indicator of progress, and, as such, trajectories can ``spin" around level sets of $V$ without changing its value. With this in mind, there is quite a great deal of freedom in how the state trajectory from a controller may evolve, even for the exact same graph of $V$ over time. 

These numerical experiments show the empirical effects of deviating from the Christoffel-consistent Coriolis matrix. Our subsequent section provides the tools that will allow us to construct it for practical systems of interest.

\section{Changes/Reduction of Coordinates}
\label{sec:induce}

This section considers how the definition of the Coriolis matrix can be transformed under a change of basis vector fields, or when the system dynamics are restricted to a submanifold. The motivation for the latter stems from the fact that general mechanical systems are often modeled as interconnected rigid bodies.
For a single rigid body, it is well-understood how one can construct the Riemannian connection \cite{bullo1999tracking} or other metric-compatible connections. Thus, it would be desirable to understand how we may derive a connection for a complex multi-body system from the connections for the bodies themselves, and to understand the properties of connections that result. Again, the opening Fig.~\ref{fig:summary} diagrams the main results as a roadmap for what follows.

\subsection{Induced Metric and Connection}

Consider a submanifold $\Qb \subseteq \Q$ with generalized speeds $\overline{v} \in \R^m$ along some frame vector fields $\overline{X}_1,\ldots,\overline{X}_m \in \X(\Qb)$.  At each $q\in \Qb$ we can write $\overline{X}_j(q) = A^i_j(q) X_i(q)$ for some coefficients $A^i_j(q)$ and arrange these coefficients into a matrix $A(q) \in \mathbb{R}^{n\times m}$ such that $A^i_j$ denotes the $i$-th row and $j$-column of the matrix $A$.
We note that since any tangent vector in $T_q \Qb$ can be written as $X_i v^i = \overline{X}_j \overline{v}^j$, we also have:
\[
\overline{X}_j \overline{v}^j = (A^i_j X_i) \overline{v}^j  = X_i v^i 
\]
and so it follows that $v^i = A^i_j \overline{v}^j$, i.e.,
\begin{equation}
v = A \overline{v}
\label{eq:relate_speeds}
\end{equation}


With these definitions, the induced metric $\Gb$ on $\Qb$ \cite{bullo2004geometric} is represented by a new mass matrix $\overline{H}$ with entries:
\[
\Hb_{ij} = \l \Xb_i , \Xb_j \r = A^k_i H_{k\ell}  A^\ell_j
\]
and, as a result: $\Hb = A\T H A$. 

A natural question then regards how we can achieve something similar for inducing Coriolis factorizations. To do so, we consider first inducing a connection $\nablab$ on $\Qb$ from a connection $\nabla$ on $\Q$ in the following manner. Let us take a curve on $\Qb$ and a vector field tangent to $\Qb$ along it. Then, viewing both as associated with $\Q$, we can use the connection on $\Q$ to take a covariant derivative via $\nabla$ at each point along the curve, and project the result back to the local tangent for $\Qb$. In other words, for any $X,Y \in \X(\Qb)$ we define our induced connection $\nablab$ so that at any $q\in \Qb$:
\[
\nablab_{X} Y(q) = {\rm proj}_{T_q \Qb} \nabla_X Y(q)
\]
where the projection of any vector $W \in T_q \Q$ is given by:
\[
{\rm proj}_{T_q \Qb} W = \argmin_{Z \in T_q \Qb} \l Z-W, Z-W\r_q
\]
Given our extrinsic construction (from Sec.~\ref{sec:background}) for the Riemannian connection, it follows that inducing a connection on $\Qb$ from the Riemannian connection $\nabla^\star$ on $\Q$ gives the Riemannian connection on $\Qb$, i.e., that
\begin{equation}
\nablab_{X}^\star Y(q) = {\rm proj}_{T_q \Qb} \nabla_X^\star Y(q)
\label{eq:induced_reim_connection}
\end{equation}

\subsection{Coriolis Factorizations for Constrained Systems}
We now consider $\Qb$ to be the configuration space of a constrained mechanism, with $\Q$ the configuration manifold when some constraints are freed. 

\begin{theorem}
\label{thm:transform}
Consider any connection $\nabla$ on $\Q$ with associated Coriolis matrix $C$. The Coriolis matrix associated with the induced connection on $\Qb$ is given by 
\begin{equation}
\overline{C} = A\T C A + A\T H \dot{A}
\label{eq:transform_C}
\end{equation}
where $\dot{A}_{j}^i =  \overline{v}^k \Lie{\overline{X}_k}{A^i_{j}} $
\end{theorem}
\begin{proof}
We first note that the projection for the induced connection means that the difference $\nabla_{\Xb_j} \Xb_k - \nablab_{\Xb_j} \Xb_k$ will be orthogonal to $T_q \Qb$.  As such, for any $i=1,\dots, m$:
\[
\l \Xb_i , \nabla_{\Xb_j} \Xb_k - \nablab_{\Xb_j} \Xb_k \r = 0.
\]
This means that the coefficients for the corresponding induced connection on $\Qb$ are given by $\Gammab_{ijk} = \l \Xb_i, \nabla_{\Xb_j} \Xb_k \r$. As a result, we  have  $\Cb_{ij} = \l \Xb_i, \nabla_{\Xb_k} \Xb_j \r \vb^k = \l \Xb_i, \nabla_{\qd} \Xb_j \r$ as a starting point. Proceeding with computations:
\begin{align*}
\Cb_{ij} &= \l \Xb_i, \nabla_{\qd} \Xb_j \r \\ 
         &= \l A_i^k X_k, \nabla_{\qd} A_j^\ell X_\ell \r \\
         &= \l  A_i^k  X_k, \dot{A}_j^\ell X_\ell + A_j^\ell \nabla_{\qd} X_\ell  \r  \\
         &= \l  A_i^k  X_k, \dot{A}_j^\ell X_\ell + A_j^\ell \nabla_{X_w} X_\ell \,  v^w \r  \\
         &= A_i^k H_{k \ell} \dot{A}_j^\ell + A_i^k \Gamma_{kw\ell} v^w A^\ell_j \\
         &= A_i^k H_{k \ell} \dot{A}_j^\ell + A_i^k C_{k \ell} A^\ell_j
\end{align*}
which implies that $ \ \Cb = A\T C A + A\T H \dot{A}$.
\end{proof}

As the main theoretical results of the paper, we show that the passivity, geodesic, and Riemannian connection properties of the Coriolis matrix on $\Q$ are inherited by the Coriolis matrix definition on $\Qb$ via the transformation law \eqref{eq:transform_C}.

\begin{proposition}
Consider a choice of generalized velocities on $\Q$ and another on $\Qb \subseteq \Q$ related by \eqref{eq:relate_speeds}. Under the  transformation law \eqref{eq:transform_C}, if $\dot{H}-2C$ is skew-symmetric then $\dot{\overline{H}}- 2\overline{C}$ is skew-symmetric.
\end{proposition}

\begin{proof} By computation.
\begin{align*}
\dot{\overline{H}}- \overline{C}-\overline{C}\,{}\T &= (\dot{A}\T H A + A\T \dot{H} A + A\T H \dot{A}) \\ 
& ~~~ - (A\T C A + A\T H \dot{A}) -  (A\T C\T A + \dot{A}\T H A) \\
&= A\T ( \dot{H} - C-C\T ) A = 0
\end{align*}
\end{proof}

\begin{proposition}
Consider a choice of generalized velocities on $\Q$ and another on $\Qb \subseteq \Q$ related by \eqref{eq:relate_speeds}. Under the transformation law \eqref{eq:transform_C} if $C$ gives the correct dynamics on $\Q$, then $\overline{C}$ does on $\Qb$.
\end{proposition}

\begin{proof}
By computation. We use the property that the Coriolis matrix $C^\star$ associated with the Riemannian connection gives the correct dynamics on $\Q$.
\begin{align*}
\overline{C}\, \overline{v} &= (A\T C A + A\T H \dot{A} ) \overline{v}  =(A\T C^\star A + A\T H \dot{A} ) \overline{v} = \overline{C}^\star \, \overline{v}
\end{align*}
where the equivalence of $A\T C^\star A + A\T H \dot{A}$ and $\overline{C}^\star$ is given in the following.
\end{proof}

\begin{proposition}
\label{prop:Cstar_transform}
Consider a choice of generalized velocities on $\Q$ and another on $\Qb \subseteq \Q$ related by \eqref{eq:relate_speeds}. Under the transformation law \eqref{eq:transform_C}, if $C=C^\star$ for $\Q$ then $\overline{C}= \overline{C}^\star$ for $\Qb$.
\end{proposition}

\begin{proof}
The result follows from \eqref{eq:induced_reim_connection} and Thm.~\ref{thm:transform}.
\end{proof}

\begin{remark}
A special case of Prop.~\ref{prop:Cstar_transform} was proved in \cite[Thm.~1]{echeandia2021numerical} for a change of generalized coordinates.
\end{remark}

\section{Computation}
\label{sec:computation}

Given a complex mechanical system, the previous results say little about how to efficiently compute Coriolis matrices with the desired properties. In this section, we consider the case of a system of rigid bodies connected by joints. In Section \ref{sec:comp_via_jaco}, we take a maximal coordinates approach and employ the results of the previous section to map a Coriolis matrix in the maximal choice of generalized speeds down to one for a minimal choice. In Sec.~\ref{sec:recursive}, we consider the special case of a tree-structure system where this mapping gives rise to an efficient algorithm. Sec.~\ref{sec:closed-chain} then discusses how systems with local loop closures (e.g., from actuation submechanisms) can likewise make use of this 
\ifincludechristoffel
algorithm, and how we can derive from it a method for directly computing the Christoffel symbols. 
\else
algorithm.
\fi
\ifincludeappendices
Beyond the Coriolis matrix, these methods can be used to compute many different quantities in passivity-based adaptive control, which are offered in Appendix~\ref{app:adaptive}.
\fi
Finally, we discuss how Coriolis matrices can be computed via differentiation of the Coriolis/centripetal terms in Sec.~\ref{sec:diff}.

\subsection{Computation Via Jacobians}
\label{sec:comp_via_jaco}

\renewcommand{\v}{\mathbf{v}}
\renewcommand{\u}{\mathbf{u}}
\newcommand{\I}{\mathbf{I}}
\newcommand{\crm}{\times}
\newcommand{\crf}{\times^\#}
\newcommand{\icrf}{\,\overline{\!\times\!}{}^\#\,}
\newcommand{\J}{\mathbf{J}}
\renewcommand{\XM}[2]{ {}^{#1} \mathbf{X}_{#2}}
\renewcommand{\bPhi}{\boldsymbol{\Phi}}
\newcommand{\Cbody}{\mathbf{C}}
\renewcommand{\X}{\mathbf{X}}
\newcommand{\B}{\mathbf{B}}
\renewcommand{\vv}{\boldsymbol{v}}
\newcommand{\oomega}{\boldsymbol{\omega}}
\newcommand{\Tmat}{\mathbf{T}}
\newcommand{\eeta}{\boldsymbol{\eta}}
\renewcommand{\vh}{\boldsymbol{h}}
\renewcommand{\vJ}{\boldsymbol{J}}
\newcommand{\Ibar}{\,\overline{\!\boldsymbol{I}}}
\newcommand{\vSigma}{\bm{\mathit{\Sigma}}}

\newcommand{\gv}{\mathsf{v}}
\newcommand{\gu}{\mathsf{u}}
\newcommand{\ga}{\mathsf{a}}

\newcommand{\gI}{\mathsf{I}}
\newcommand{\gcrm}{\times}
\newcommand{\gcrf}{\times^\#}
\newcommand{\gicrf}{\,\overline{\!\times\!}{}^\#\,}
\newcommand{\gJ}{\mathsf{J}}
\newcommand{\gXM}[2]{ {}^{#1} \mathsf{X}_{#2}}
\newcommand{\gbPhi}{\mathsf{\Phi}}
\newcommand{\gCbody}{\mathsf{C}}
\newcommand{\gX}{\mathsf{X}}
\newcommand{\gB}{\mathsf{B}}

\subsubsection{Single Rigid Body} Let us first consider a single rigid body whose configuration manifold $\Q = SE(3)$ is represented by all matrices of the form:
\[
\Tmat = \begin{bmatrix} \boldsymbol{R} & \boldsymbol{p} \\ 0_{3\times 1} & 1 \end{bmatrix}
\]
where $\boldsymbol{R} \in SO(3)$ is a rotation matrix specifying the orientation of a body-fixed frame relative to an earth-fixed frame, and $\boldsymbol{p} \in \R^3$ spots the origin of that body frame relative to the earth frame. In that case, we consider the body twist $\v \in \mathbb{R}^6$ as the vector of generalized speeds where
\[
\v = \begin{bmatrix} \oomega \\ \vv \end{bmatrix}
\]
with $\oomega \in \R^3$ the angular velocity of the body (expressed in the body-fixed frame), and $\vv\in \R^3$ the velocity of the origin of that frame (also expressed in the body frame). We can identify $\v\in\mathbb{R}^6$ with a left-invariant vector field on $SE(3)$ \cite{bullo2004geometric} according to
\[
X(\Tmat) = \Tmat \underbrace{ \begin{bmatrix} (\oomega \times) & \vv \\ 0_{3 \times 1} & 0 \end{bmatrix}}_{ \triangleq \v^\wedge}
\]
where the hat map $^\wedge$ promotes the components of $\v$ to an element in the Lie algebra $\se{3}$. We consider the spatial inertia $\I$ in the body frame \cite{featherstone2014rigid}
\[
\I = \begin{bmatrix} \Ibar & (\vh \times) \\ 
                     (\vh \times)\T & m\bone  \end{bmatrix}
\]
where $\Ibar \in \R^{3\times3}$ is the rotational inertia about the origin of the body frame, $m$ the body mass, and $\vh = m \boldsymbol{c}$ the first mass moment with $\boldsymbol{c}$ the vector to the center of mass of the body \cite{featherstone2014rigid}. We employ this matrix to set a left-invariant metric $\mathbb{G}$, which we form via:
\[
\l \Tmat \v^\wedge, \Tmat \v^\wedge \r_{\Tmat} = \l \v^\wedge, \v^\wedge \r_{\mathbf{1}} = \v\T \,\I\, \v
\]

The Koszul formula can then be used to evaluate the Riemannian connection for this metric. We consider $Y(\Tmat) = \Tmat \u^\wedge$ as another left-invariant vector field for some $\u \in \mathbb{R}^6$. Following the results of \cite[Eqn.~(8)]{bullo1999tracking} and \cite[App.~2.F]{arnol2013mathematical}, 
\[
\nabla^\star_X Y = \Tmat \Big[ \, \overset{\mathfrak{g}}{\nabla}_\v \u \, \Big]^\wedge 
\]
where the map $\overset{\mathfrak{g}}{\nabla} : \mathbb{R}^6 \times \mathbb{R}^6 \rightarrow \mathbb{R}^6$ satisfies
\begin{equation}
\I \, \overset{\mathfrak{g}}{\nabla}_{\v} \u = \frac{1}{2} \left[ \I (\v \crm) + (\v \crf) \I  + (\I \v)\icrf\right] \u 
\label{eq:RiemanninConnectionRigidBody}
\end{equation}
and where
\[
(\v \crm) = \begin{bmatrix} \oomega \times & 0 \\ \vv \times & \oomega \times \end{bmatrix}
\]
is called a spatial-cross product matrix \cite{featherstone2014rigid} or equivalently the (little) adjoint matrix associated with $\v$ \cite{bullo2004geometric}. The quantity $(\v \crf)$ is defined by $ (\v \crf) = - (\v \crm)\T$, and  $(\f \icrf) \in \R^{6 \times 6}$ for $\f \in \mathbb{R}^6$ is the unique matrix such that $(\f \icrf) \v = (\v \crf) \f$. As a result, the quantity 
\begin{equation}
\Cbody = \frac{1}{2} \left[ \I (\v \crm) + (\v \crf) \I  + (\I \v)\icrf\right]
\label{eq:Cbody}
\end{equation}
thus gives the Christoffel-consistent factorization for a single rigid body when using its body twist for generalized speeds. 

\subsubsection{System of Rigid Bodies} Now, let us consider a system of $N$ bodies connected with joints imposing relative motion constraints. In maximal coordinates, let us choose generalized speeds $v = [\v_1; \ldots  ; \v_N]$. 

Defining $\Cbody_i = \frac{1}{2} \left[ \I_i (\v_i \crm) + (\v_i \crf) \I_i  + (\I_i \v_i)\icrf\right]$, in this set of maximal coordinates, the mass matrix and Christoffel-consistent Coriolis matrix are given by:
\begin{align*}
H &= {\rm BlockDiag}(\I_1, \ldots, \I_N)\\  C^\star &= {\rm BlockDiag}( \Cbody_1, \ldots, \Cbody_n) 
\end{align*}

Suppose we can choose minimal generalized speeds $\overline{v}$. 
We then consider Jacobians $\J_i$ for each body $i$ such that $\v_i = \J_i \overline{v}$. We can stack these Jacobians into a matrix $A = [\J_1 ; \cdots ; \J_N]$ such that $v = A \overline{v}$. Then, for the constrained mechanism, we immediately have
\begin{equation}
\Cb^\star = A\T C^\star A + A\T H \dot{A}
\label{eq:cstar_mechanical}
\end{equation}
via Proposition \ref{prop:Cstar_transform}.

\newcommand{\vjb}[1]{\overline{\tt v}{}_{{#1}}}
\newcommand{\qjb}[1]{\overline{\tt q}{}_{{#1}}}
\newcommand{\qj}[1]{{\tt q}{}_{{#1}}}

\subsection{Recursive Computation for Open-Chain Systems}
\label{sec:recursive}

We now consider how to efficiently numerically compute \eqref{eq:cstar_mechanical} for a collection of bodies (e.g., a humanoid robot) free of kinematic loops. 

We consider the $N$ bodies connected via a set of $N$ joints. We can choose minimal generalized speeds $\overline{v}$ as a collection of generalized speeds ($\vjb{1}, \ldots, \vjb{N}$) for each joint. Bodies are numbered $1$ through $N$ such that the predecessor $p(i)$ of body $i$ (toward the root) satisfies $p(i) < i$ \cite{featherstone2014rigid}. When a joint $i$ is on the path from body $j$ to the root, we denote that $i\preceq j$.

Velocities of neighboring bodies can then be related by
\[
\v_i = \XM{i}{p(i)} \v_{p(i)} + \bPhi_i(\qjb{i}) \, \vjb{i}
\]
where $\XM{i}{p(i)}$ is the spatial transform that changes frames from $p(i)$ to $i$ (equivalently given as the Adjoint matrix $\XM{i}{p(i)} = {\rm Ad}_{{}^i \Tmat_{p(i)}}$), and where the quantity $\bPhi_i$ describes the free modes for joint $i$ in its local frame \cite{featherstone2014rigid}. In what follows, we will use $\qjb{i}$ to denote the configuration for the $i$-th joint. 

Suppose we have a joint $j$ that appears earlier in the tree (toward the root) before body $i$ (i.e., $j\preceq i$). The non-zero block of columns corresponding to joint $j$ in the body Jacobian $\J_i$ are given by $\XM{i}{j} \bPhi_j$. Collecting these entries for all bodies and joints, we can express all the body velocities together as:
\[
v = \XM{}{} \bPhi \overline{v}
\]
where $\bPhi = {\rm BlockDiag}(\bPhi_1, \ldots , \bPhi_N)$ and $\XM{}{}$ is a block matrix with sparsity pattern depending on the connectivity of the mechanism with block $i,j$ given by $\XM{i}{j}$ when $j \preceq i$ and zero otherwise. For example, if the tree has predecessor array $p=[0, 1, 1 ,3]$, then this matrix would take the form
\[
\XM{}{} = \begin{bmatrix} \XM{1}{1} & 0 & 0 & 0 \\ \XM{2}{1} & \XM{2}{2} & 0 & 0 \\ \XM{3}{1} & 0 & \XM{3}{3} & 0 \\ \XM{4}{1} & 0 & \XM{4}{3} & \XM{4}{4} \end{bmatrix}
\]
As a result of these relationships, we have $A = \XM{}{} \bPhi$. 

We proceed to consider the time rate of change of $A$. The rate of change in a block of $\XM{}{}$ is given by \cite{featherstone2014rigid}:
\[
\frac{{\rm d}}{{\rm d} t} \XM{i}{j}   = \XM{i}{j} (\v_j \times)- (\v_i \times) \XM{i}{j}
\]
such that if we define $(v\times) = {\rm BlockDiag}( \v_1\times, \ldots, \v_N\times )$ we have:
\[
\dot{\XM{}{}} = \XM{}{} (v \crm) - (v\crm) \XM{}{}   
\]
In line with Featherstone's notation \cite{featherstone2014rigid}, we denote $\ddt \bPhi_i = \mathring{\bPhi}_i$ in local coordinates and denote $\dot{\bPhi} = (v \crm) \bPhi + \mathring{\bPhi}$ which gives the rate of change in the joint axes due to their motion inertially ($(v \crm) \bPhi$), as well as in local coordinates $(\mathring{\bPhi})$\footnote{The closed dot can be interpreted as the derivative being taken in the ground coordinates, i.e., $\dot{\bPhi}_i = \XM{i}{0} \left( \ddt  \XM{0}{i} \bPhi_i \right)$.}.  The final result is that we have:
\[
\dot{A} = \dot{\XM{}{}} \bPhi + \XM{}{} \mathring{\bPhi} = \X \dot{\bPhi} - (v\times) \X \bPhi
\]

Applying our transformation result from the previous section:
\begin{align*}
\Cb^\star &= A\T C^\star A + A\T H \dot{A} \\
    &= \bPhi\T \X\T C^\star \X \bPhi + \bPhi\T \X\T H \X \dot{\bPhi} - \bPhi\T \X\T H (v\crm) \X \bPhi
\end{align*}
Then, grouping the first and third terms, and considering $B = C^\star - H (v\crm)$ we have that
\begin{align}
\label{eq:formofCstar}
\Cb^\star &= \bPhi\T \, \X\T\, B \, \X \, \bPhi + \bPhi\T \, \X\T \, H \, \X \, \dot{\bPhi}
\end{align}
The quantity $B$ takes the form $B = {\rm BlockDiag}(\B_1, \ldots, \B_N)$
where, 
\begin{equation}
\B_i = \Cbody_i - \I_i (\v_i \crm) = \frac{1}{2} \left[ (\v_i \crf) \I_i  + (\I_i \v_i)\icrf - \I_i (\v_i \crm)\right]
\label{eq:NS_factorization} 
\end{equation}
Each of these terms gives a body-level factorization of the Coriolis terms that originally appeared in \cite{Niemeyer90,niemeyer1991performance}
\ifincludeexapndedBD
\else
.
\fi

Remarkably, the formula for $\Cb^\star$ in \eqref{eq:formofCstar} is equivalent to \cite[Eq.~18]{echeandia2021numerical}, which thus implies that the main algorithm of that paper \cite[Alg.~1]{echeandia2021numerical} will compute $\Cb^\star$ for any open-chain mechanism with generic joint models. This added generality does not affect the $\mathcal{O}(Nd)$ computational complexity of the algorithm, where $d$ is the depth of the kinematic connectivity tree. The theoretical developments in \cite{echeandia2021numerical} had used previous adaptive control results \cite{niemeyer1991performance} to show that the algorithm gave $\Cb^\star$ when working with mechanisms comprised of revolute, prismatic, and helical joints. The more general theory herein justifies that this favorable property of the algorithm generalizes to a considerably larger class of mechanisms (e.g., with spherical joints, universal joints, floating bases, etc.). Implementations of this algorithm are available in MATLAB \cite{wensingweb} and in recent versions of the popular C/C++ package Pinocchio \cite{carpentier2019pinocchio} (see, e.g., discussion at \cite{pinocchio_discussion}). 

Note that this algorithm then also gives a direct method for computing generalized Christoffel symbols according to $\Gammab^\star_{ijk}(q) = [\Cb^\star(q, e_j)]_{ik}$ where $e_j$ is the $j$-th unit vector. As such, the algorithm can be used to evaluate all generalized Christoffel symbols in time $\mathcal{O}(N^2 d)$. This capability would, for example, enable evaluating covariant derivatives of the Riemannian connection for high degree-of-freedom mechanical systems. 
\ifincludechristoffel
We provide a dedicated algorithm to compute the generalized Christoffel symbols $\Gammab^\star_{ijk}$ following the next section that improves this complexity to $\mathcal{O}(N d^2)$. 
\else
\fi 

\begin{remark}
The factorization of Niemeyer and Slotine \cite{Niemeyer90, niemeyer1991performance} in \eqref{eq:NS_factorization} is closely linked with the Coriolis factorization for a single-rigid body when adopting its spatial twist for its generalized speeds. That is, if $\v$ is the body twist in a body frame $b$, and $\overline{\v} = \XM{0}{b} \v$ gives the spatial twist in an earth frame 0, then the Christoffel-consistent factorization in the $\overline{\v}$ coordinates is:
\[
\Cb = \XMT{b}{0} \B \XM{b}{0}
\]
where $\B = \frac{1}{2} \left[ (\v \crf) \I  + (\I \v)\icrf - \I (\v \crm)\right]$. Correspondingly, if we consider two {\bf right}-invariant vector fields 
$X(\Tmat) = \overline{\v}^\wedge \, \Tmat$ and $Y(\Tmat) = \overline{\u}^\wedge \, \Tmat$, it follows that:
\[
\nabla^\star_X Y = \Big[ \overset{\bar{\mathfrak{g}}}{\nabla}_{\overline \v} \, \overline{\u} \Big]^\wedge \, \Tmat
\] 
where, analogous to \eqref{eq:RiemanninConnectionRigidBody}, the map $\overset{\bar{\mathfrak{g}}}{\nabla} : \mathbb{R}^6 \times \mathbb{R}^6 \rightarrow \mathbb{R}^6$ satisfies
\begin{equation}
\overline{\I}  \, \overset{\bar{\mathfrak{g}}}{\nabla}_{\overline \v}\, \overline \u = \frac{1}{2} \left[  (\overline \v \crf) \overline \I  + (\overline \I \overline{\v} )\icrf -\overline \I ( \overline \v \crm) \right] \overline{\u} = \Cb \overline{\u}
\end{equation}
where $\overline{\I} = \XMT{b}{0} \, \I \, \XM{b}{0}$.
\end{remark}

\subsection{Efficient Computation for Closed-Chain Systems}
\label{sec:closed-chain}

For closed-chain systems, a few options can be considered to compute the Coriolis matrix. The first is to adopt a spanning tree for the rigid-body system \cite{featherstone2014rigid} and project its Coriolis matrix to minimal coordinates via \eqref{eq:transform_C}. The second is to employ constraint embedding approaches \cite{jain2012multibody} that group bodies involved in local loop closures into aggregate bodies, so that the topology of these groups becomes a tree. 

\subsubsection{Spanning Tree Method for $C$}

 In this case, one would first employ \cite[Alg.~1]{echeandia2021numerical} to obtain the Coriolis matrix $C_{\rm spanning}$ for a spanning tree of the closed-chain mechanism \cite{featherstone2014rigid}. To project to the minimal coordinates, one would then construct a mapping $v_{\rm spanning} = A(q) \, v_{\rm minimal}$ from a set of minimal velocities $v_{\rm minimal}$ to the spanning tree velocities $v_{\rm spanning}$. Then, one could use the relationship $C_{\rm minimal} = A\T C_{\rm spanning} A + A\T H_{\rm spanning} \dot{A}$ to obtain a Coriolis matrix with desired properties in the minimal coordinates. 

\subsubsection{Constraint Embedding Method for $C$}

An alternative approach is to start with the kinematic connectivity graph of the system and group bodies involved in local loop closures. This process yields tree-structured graphs for many important classes of constraints \cite{jain2012multibody,kumar2020analytical,chignoli2023propagation}. The resulting tree models organize bodies into sets (referred to as "aggregated bodies" in prior work, but called "clusters" here for smoothness of prose), such that the velocity of cluster $k$ is given by 
\cite{chignoli2023propagation}:
\[
\gv_k = \begin{bmatrix} \v_{k_1} \\ \vdots \\ \v_{k_{n_k}} \end{bmatrix}
\]
where $k_1, \ldots, k_{n_k}$ denote the indices of the bodies in the cluster. If cluster $j$ precedes cluster $k$ in the tree, then when the joints in between them are locked, their velocities are related by 
$
\gv_k = \gXM{k}{j} \, \gv_j
$, 
where the matrix $\gXM{k}{j}$ is a block matrix of spatial transforms, with one non-zero transform per block row. The matrix can again be shown to satisfy the property:
\[
\ddt \gXM{k}{j}  = \gXM{k}{j} (\gv_j \times) - (\gv_k \times) \gXM{k}{j} 
\] 
where we define:
\[
(\gv_k\times) = {\rm BlockDiag}\left( (\v_{k_1} \times), \ldots, (\v_{k_{n_k}} \times)  \right)
\]
Further, the velocity of cluster $k$ can be related to its parent via some $\gbPhi_k(\qjb{k})$ \cite{jain2012multibody,chignoli2023propagation} so that:
\[
\gv_k = \gXM{k}{p(k)} \gv_{p(k)} + \gbPhi_k(\qjb{k})\, \vjb{k}
\]
which, again, allows for an analogous development to the previous section such that we write the map from minimal velocities $\overline{v}$ to all rigid body velocities as:
\[
v = \gXM{}{} \gbPhi \overline{v}
\] 
where $\gXM{}{}$ is constructed analogously for clusters as it was for bodies in the previous section, and $\gbPhi = {\rm BlockDiag}(\gbPhi_1, \ldots, \gbPhi_{N_c})$ with $N_c$ the number of clusters. 
\newcommand{\Tr}[1]{^{\!\top_{\!#1}}}

So, again, we have:
\begin{align}
\Cb^\star &= \gbPhi\T \, \gX\T\, B \, \gX \, \gbPhi + \gbPhi\T \, \gX\T \, H \, \gX \, \dot{\gbPhi}
\label{eq:CbStar}
\end{align}
This allows us to immediately generalize many formulas and the main algorithm from \cite{echeandia2021numerical} to the case of working with clusters, which is given in Alg.~\ref{alg:cor}. Overall, this algorithm computes each non-zero block $(i,j)$ of $\Cb$ to as
\newcommand{\mx}[1]{\lceil#1 \rceil}
\[
\Cb_{ij}^\star = \gbPhi_i\T \, \gXM{\mx{ij}}{i}\T \gI_{\mx{ij}}^C \gXM{\mx{ij}}{j} \dot{\gbPhi}_j + \gbPhi_i\T \, \gXM{\mx{ij}}{i}\T \gB_{\mx{ij}}^C \gXM{\mx{ij}}{j} \gbPhi_j
\]
where $\mx{ij} = {\rm max}(i,j)$ when clusters $i$ and $j$ are related (i.e., when one is a descendant of the other in the connectivity tree). The quantities
\begin{align}
\gI_{k}^C &= \sum_{\ell \succeq k} \gXM{\ell}{k}\T \gI_\ell \gXM{\ell}{k} & 
~~~\gB_{k}^C &= \sum_{\ell \succeq k} \gXM{\ell}{k}\T \gB(\gI_\ell, \gv_\ell) \gXM{\ell}{k} 
\end{align}
give the composite inertia and composite Coriolis matrices for the cluster subtree rooted at cluster $k$ where
\[
\gB(\gI, \gv) = \frac{1}{2} \left[ (\gv \crf) \gI  + (\gI \, \gv)\icrf - \gI (\gv \crm)\right] 
\]
gives a generalization of the previous formula for $\vB_i$ when working with cluster quantities.

\newcommand{\gF}{\mathsf{f}}
\newcommand{\gPhiDot}{\dot{\gbPhi}}

\begin{algorithm}[t]
    \caption{Torsion-Free Coriolis Matrix Algorithm}
    \setstretch{1.2}
    \begin{algorithmic}[1]
    \REQUIRE $q,\,\overline{v}$
    \STATE $\gv_0 = \bzero$
    \FOR{$i=1$ to $N$}
    \STATE $\gv_i = \gXM{i}{p(i)}\,\gv_{p(i)}+\gbPhi_i \,\vjb{i}$ 
    \STATE $ \dot{\gbPhi}_i = (\gv_i \times) \gbPhi_i + \mathring{\gbPhi}_i(\qjb{i}, \vjb{i})$ 
    \STATE $\gI_i^C = {\rm diag}(\vI_{i_1}, \ldots, \vI_{i_{n_i}})$ 
    \STATE $\gB_i^C = \frac{1}{2}[ (\gv_i \times^*) \gI_i + (\gI_i \gv_i)\icrf-   \gI_i (\gv_i \times )]$ 
    \ENDFOR
    \FOR{$j=N$ to $1$}
    \STATE $\gF_1 = \gI_j^C\, \gPhiDot_j + \gB_j^C \, \gbPhi_j  $ 
    \STATE $\gF_2 = \gI_j^C\, \gbPhi_j$ 
    \STATE $\gF_3 = (\gB_j^C)\T \gbPhi_j$ 
    \STATE $\overline{C}_{jj} = \gbPhi{}_j\T\, \gF_1$ 
    \STATE $i=j$
    \WHILE{$p(i)>0$}
    \STATE $\gF_1\!=\!\gXM{i}{p(i)}\T\,\gF_1$;\,$\gF_2\!=\!\gXM{i}{p(i)}\T\,\gF_2$;\, $\gF_3\!=\!\gXM{i}{p(i)}\T\,\gF_3$
    \STATE $i = p(i)$ 
    \STATE $\overline{C}_{ij}^\star = \gbPhi{}_i\T\, \gF_1$ 
    \STATE $\overline{C}_{ji}^\star = (\gPhiDot{}_i\T\, \gF_2 + \gbPhi{}_i\T \gF_3 )\T$
    \ENDWHILE
    \STATE $\gI_{p(j)}^C = \gI_{p(j)}^C + \gXM{j}{p(j)}\T \, \gI_j^C \, \gXM{j}{p(j)} $\\[.5ex]
    \STATE $\gB_{p(j)}^C = \gB_{p(j)}^C + \gXM{j}{p(j)}\T \, \gB_j^C \, \gXM{j}{p(j)} $
    \ENDFOR
    \RETURN $\overline{C}^\star$
    \end{algorithmic}
    \label{alg:cor}
    \end{algorithm}

\ifincludechristoffel
\subsubsection{Generalized Christoffel Symbols with Constraint Embedding}
\label{sec:ClusterTree}
In practice, once could run this algorithm with unit vector inputs for the components of $\vb$ to obtain the components of $\Gammab$. Alternatively, we provide an efficient implementation for the generalized Christoffel symbols directly, which is given in Appendix \ref{app:christoffel} and with {\sc Matlab} source available at \cite{wensingweb_Christoffel}. 
\fi

\ifincludeappendices
\ifincludechristoffel
\subsubsection{Algorithms in Support of Adaptive Control}
\fi
More broadly, the structural form of the models for these local loop closures enables many conventional adaptive approaches to extend to this broader class of systems. Indeed, as noted in \cite{jain2012multibody}, many recursive algorithms carry over directly to the case of local loop closures modeling using constraint embedding. 
\ifincludechristoffel
The Christoffel algorithm, however, highlights a situation where this is not the case. 
\fi
Thankfully, this is the case for recursive regressor computation algorithms and many other algorithms for direct and indirect adaptive control. A short justification and supporting pseudocode is provided in Appendix \ref{app:adaptive}.
\fi


\subsection{Computation via Differentiation} 
\label{sec:diff}

In the interest of completeness, we close this section by describing a way to obtain the Christoffel-consistent Coriolis matrix via derivative calculations applied to the Coriolis and centripetal terms. When working with generalized coordinates, one can compute $C^\star$ via:
\[
[C^\star]_{ik} = \frac{1}{2} \frac{\partial [C(q,\dot{q})\dot{q}]_i }{\partial \dot{q}^k}
\]
However, this formula can break down when working with generalized speeds, since the generalized Christoffel symbols are not guaranteed to be symmetric in their (2,3) indices. 
\allowdisplaybreaks 

Indeed, the derivatives of $C(q,v)v$ w.r.t. $v$ give:
\[
\frac{\partial [C(q,v)v]_i}{\partial v^k} = 2 \, \Gamma^\star_{i(jk)} v^j 
\]
As a result, using \eqref{eq:gamma_symmetry}, we have that:
\begin{align}
C_{ik}^\star &= \frac{1}{2} \frac{\partial [C(q,v)v]_i}{\partial v^k} + \frac{1}{2} s_{ijk} v^j \\
             &= \frac{1}{2} \frac{\partial [C(q,v)v]_i}{\partial v^k} + \frac{1}{2} H_{i\ell} s_{jk}^\ell v^j \\
             &= \frac{1}{2} \frac{\partial [C(q,v)v]_i}{\partial v^k} + \frac{1}{2} H_{i\ell} X^\ell \cdot [X_j , X_k] v^j \label{eq:deriv_result}
\end{align}
where $X^1, \ldots , X^n$ gives the basis dual to $X_1, \ldots, X_n$. In this regard, one can construct the Christoffel-consistent Coriolis matrix via derivatives of the dynamics equations when working with generalized coordinates, but such a construction requires a correction factor via the structure constants when working with generalized speeds.

For example, for open-chain robots with joint models satisfying $\mathring{\bPhi}_i=0$, 
we can compute \eqref{eq:deriv_result} as follows. In this case, the structure constants $s^\ell_{jk} = X^\ell \cdot [X_j , X_k]$ are configuration independent and are only nonzero when $\ell$, $j$, and $k$ all belong to the same joint (since the motions of different joints will commute). Considering the columns of  $\bPhi_i$ as a partial basis for $\R^6$, we complete the basis via any complimentary matrix $\bPhi_i^c$ such that $[\bPhi_i ~ \bPhi_i^c] \in \R^{6\times 6}$ is full rank. Defining $[ \boldsymbol{\Psi}_i~\boldsymbol{\Psi}_i^c ] = [\bPhi_i ~ \bPhi_i^c]^{-\top}$ we have that $\boldsymbol{\Psi}_i\T \bPhi_i = \mathbf{1}$ via construction. 

We use these dual vectors $\boldsymbol{\Psi}_i$ to compute \eqref{eq:deriv_result}. 
\newcommand{\vj}[1]{ {\tt v}_{#1}}
Consider local blocks:
\[
{\mathbf K}_i(q,v) =  \mathbf{\Psi}_i\T ( \mathbf{\Phi}_i \vj{i} )\times \mathbf{\Phi}_i 
\]
where, again, $\vj{i}$ is the subvector of $v$ corresponding to the velocity components of joint $i$.
We stack these ${\mathbf{K}}_i$ matrices into a block diagonal matrix $K = {\rm BlockDiag}({\mathbf K}_1, \ldots, {\mathbf K}_N)$, which then enables implementing \eqref{eq:deriv_result} via
\begin{equation}
C^\star(q,v) = \frac{1}{2} \left[ \frac{ \partial C(q,v) v}{\partial v} + H(q) K(q,v) \right]
\label{eq:C_via_diffs}
\end{equation}
This approach might be more straightforward for those who have access to an efficient auto-differentiation package (e.g., such as is part of CasADi  \cite{andersson2019casadi}). A subsequent derivative w.r.t. $v$ can then be used to compute the generalized Christoffel symbols via auto-differentiation if needed.

\later{
\section{Applications to Passivity-Based Control}

\subsection{Passivity-Based Control of a Manipulator}

\subsection{Passivity-Based Control of a Humanoid}
}

\section{Conclusions}
\label{sec:conclude}

This paper has provided a comprehensive overview of the link between Coriolis factorizations for mechanical systems and affine connections on their configuration manifold. Analysis of the contorsion tensor (measuring the difference from the Riemannian connection) proved that there is a unique factorization satisfying the skew property when $n\le 2$, but an infinity of options when $n\ge3$. Despite this flexibility, the Coriolis matrix associated with the Riemannian connection (and therefore associated with the Christoffel symbols) inherits an additional property of being torsion-free, which we have shown can avoid excess twisting of trajectories in passivity-based control. While the presence of torsion does not affect the Lyapunov argument in this case, it leads to oscillations in the control output that could incite unmodeled high-frequency vibrations in mechanical systems. This motivation led us to address how to compute the Christoffel-consistent Coriolis matrix for constrained mechanical systems, wherein our results on induced factorizations played a central role. For open-chain mechanical systems, this Coriolis matrix can be computed with computational complexity $O(Nd)$, which promotes its use in control strategies for systems such as humanoids and quadruped robots. Extensions demonstrated how these constructions can be applied for computation of the Coriolis matrix and generalized symbols of the Riemannian connection for systems with local loop closures. These algorithms may be considered for other applications beyond passivity-based control (e.g., when taking covariant derivatives in other contexts). 


\section*{Acknowledgement}
The authors gratefully acknowledge Jared DiCarlo for 
\ifincludeconstruction
\else
observations that originally motivated Corrolary \ref{cor:antisymmetric}. 
\fi
We also thank Gianluca Garofalo and Johannes Englsberger, whose insightful suggestions reinforced the directions we were exploring and motivated us to further develop them. Specifically, we thank Gianluca for originally making us aware of \eqref{eq:RiemanninConnectionRigidBody}. Additionally, we thank Jake Welde for his stimulating discussions on this topic and Justin Carpentier for his kind assistance in updating the Coriolis computations in Pinocchio for consistency with the Christoffel symbols.  
\bibliographystyle{IEEEtran}
\bibliography{affine}

\ifincludeappendices
\appendices

\newcommand{\cF}{\mathcal{F}}
\newcommand{\cB}{\mathcal{B}}
\newcommand{\cA}{\mathcal{A}}

\newcommand{\gw}{\mathsf{w}}
\newcommand{\gW}{\mathsf{W}}
\newcommand{\gV}{\mathsf{V}}

\ifincludechristoffel

\begin{figure}
  \centering
  \includegraphics[width=\columnwidth]{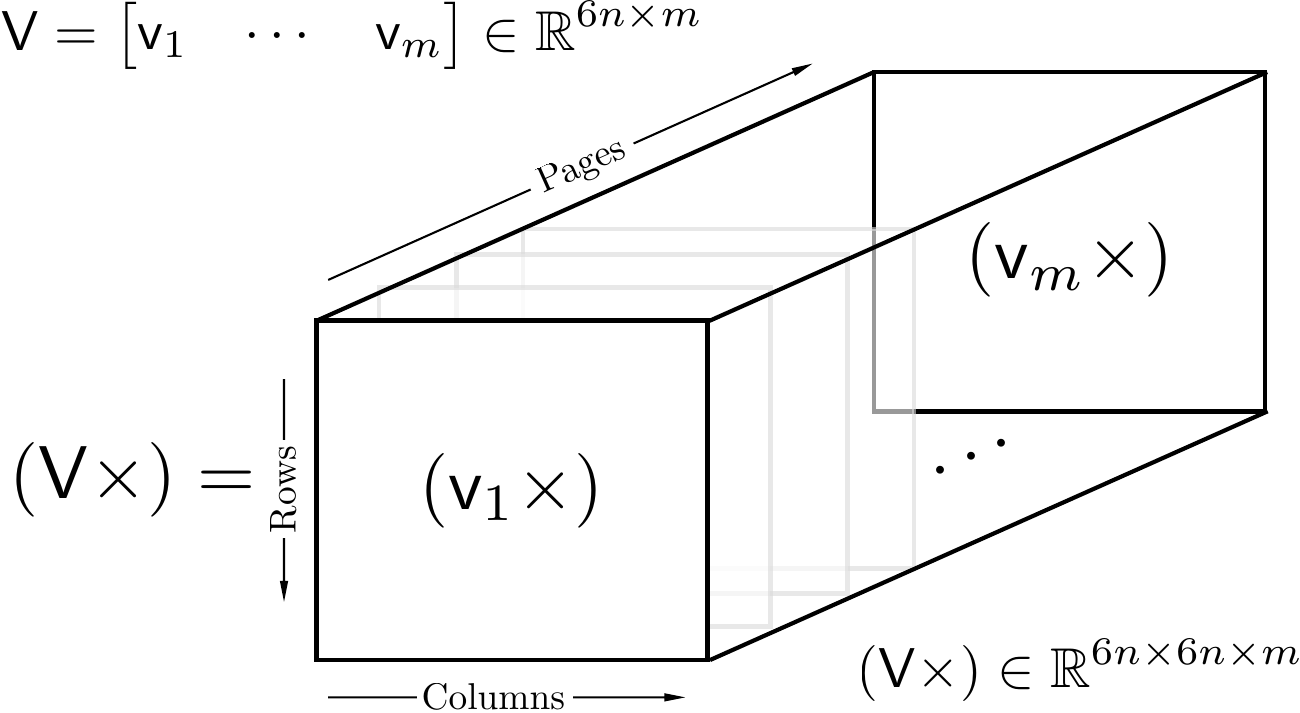}
  \caption{Storage order for the applying the cross product operator to a matrix. A cross product matrix for each column is placed along each page of the dimension 3 tensor.}
  \label{fig:Pages}
\end{figure}

\section{Generalized Christoffel Symbols with Constraint Embedding Joint Models}
\label{app:christoffel}

\begin{algorithm}[t]
\caption{Generalized Christoffel Symbols for Clusters}
\setstretch{1.2}
\begin{algorithmic}[1]
\REQUIRE $q$
\FOR{$i=1$ to $N_c$} 
\STATE $\gI_i^C = {\rm diag}(\vI_{i_1}, \ldots, \vI_{i_{n_i}})$ 
\STATE $\mathcal{D}_i = \left( \frac{\partial \mathring{\gbPhi}_i}{\partial \vjb{i}} \right)^{\!\top_{\!23}}$
\ENDFOR
\FOR{$k=N_c$ to $1$}
\STATE $\mathcal{B} = \gB( \gI^C_k, \gbPhi_k)  $ 
\STATE $\cF_3 = \gI_k^C \, \mathcal{D}_k $
\STATE $\gF_4 = \gI_k^C \, \gbPhi_k$
\STATE $j = k$
\WHILE{$j > 0$}
\STATE $\cF_1 = \cB \,\gbPhi_j$, $\cF_2 = \cB\T \, \gbPhi_j$
\STATE $i = j$
\WHILE{$i>0$}

\STATE $\cA_1 = \gbPhi_i\T \cF_1$, $\cA_2 = \left( \gbPhi_i\T \cF_2\right)^{\!\top_{\!12}} $, $\cA_3 = -(\cA_2)^{\!\top_{\!13}}$

\STATE $\Gammab_{ijk} =\cA_1 $, 
       $\Gammab_{ikj} = \left(\cA_1\right)^{\top_{\!23}}$

\STATE $\Gammab_{jik} = \cA_2$, 
        $\Gammab_{jki} = \left( \cA_2 \right)^{\top_{\!23}}$

\STATE $\Gammab_{kij} = \cA_3$, $\Gammab_{kji} = \left( \cA_3 \right)^{\top_{\!23}}$

\IF{$i=j$}
\STATE $\Gammab_{kii} = \cA_3 + \gF_4\T \mathcal{D}_i$
\ENDIF

\IF{$j = k$}
\STATE $\Gammab_{ijj} = \cA_1 + \gbPhi_i\T \cF_3$ 
\STATE $\cF_3 = \gXM{i}{p(i)}\T \cF_3$
\ENDIF

\STATE $\cF_1 = \gXM{i}{p(i)}\T \cF_1$, $\cF_2 = \gXM{i}{p(i)}\T \cF_2$

\STATE $i = p(i)$
\ENDWHILE
\STATE $\cB = \gXM{j}{p(j)}\T  \, \cB\, \gXM{j}{p(j)}$, $\gF_4 = \gXM{j}{p(j)}\T \gF_4$
\STATE $j = p(j)$
\ENDWHILE
\STATE $\gI_{p(k)}^C = \gI_{p(k)}^C + \gXM{k}{p(k)}\T  \, \gI_k^C\, \gXM{k}{p(k)}$
\ENDFOR
\RETURN $\Gammab$
\end{algorithmic}
\label{alg:Gamma}
\end{algorithm}

The development for obtaining the generalized Christoffel symbols when working with constraint embedding joint models makes use of rank 3 tensors (with the dimensions sequentially identified as rows, columns, and pages as in Fig.~\ref{fig:Pages}). We use the symbol $(\cdot)\Tr{12}$ to denote transposing the first two dimensions (rows and columns), and likewise for $(\cdot)\Tr{13}$ and $(\cdot)\Tr{23}$. Deriving the algorithm for the Christoffel symbols from \eqref{eq:CbStar} follows the same conceptual approach as in \cite{echeandia2021numerical}, but with the extra bookkeeping and conventions required for working with tensors. For example, multiplication between a matrix and a rank 3 tensor is given by:
\[
(A \mathcal{B} )_{ijk} = \sum_\ell A_{i\ell} \mathcal{B}_{\ell jk},  
\]
with multiplication between a rank 3 tensor and a matrix given by:
\[
(\mathcal{A} B)_{ijk} = \sum_\ell \mathcal{A}_{i \ell k} B_{\ell j}
\]
so that, in either case, multiplication occurs page-by-page for the tensor. We adopt the same conventions as in \cite{singh2024second} where for any $\gbPhi_i \in \mathbb{R}^{6 n_i \times d_i}$ we construct $(\gbPhi_i \times) \in \mathbb{R}^{6 n_i \times 6 n_i \times d_i}$ by considering the usual cross product $(~\times)$ for each column of $\gbPhi_i$ and stacking the results along pages (Fig.~\ref{fig:Pages}). 

To arrive at the computational structure in Algo.~\ref{alg:Gamma}, we employ the following identities when $j$ and $k$ are related:
\begin{align}
\frac{\partial \gB_k^C}{\partial \vjb{j} } &= \sum_{\ell \succeq \mx{jk} } \gXM{\ell}{k}\T \, \gB\left( \gI_\ell, \, \gXM{\ell}{j} \gbPhi_j \right) \, \gXM{\ell}{k} \\
&= \gXM{\mx{jk}}{k}\T \, \gB\left(\gI_{\mx{jk}}^C,\,  \gXM{\mx{jk}}{j} \gbPhi_j \right) \, \gXM{\mx{jk}}{k}
\end{align}
and
\begin{align}
\frac{\partial \dot{\gbPhi}_k}{\partial \vjb{j} } = \gbPhi_j \times \gbPhi_k + \frac{\partial \mathring{\gbPhi}_k }{\partial \vjb{j} }~~~ (j\preceq k)
\end{align}
where partial derivatives w.r.t. the different components of $\vjb{j}$ are stacked along pages.

In arranging the computations for recursive implementation, we employ that for matrices $\gV$ and $\gW$ and inertia $\gI$ with matched row-dimensions, we have \cite[Table 3]{singh2024second}:
\begin{align}
\gB( \gI, \gV )\Tr{12} \,\gW &= 
( -\gB(\gI, \gW)\Tr{12}\, \gV )\Tr{23} \\
\gB( \gI, \gV) \gW &= \left( \gB(\gI,\gW)\gV + \gI \, (\gW\times ) \gV \right)\Tr{23} 
\end{align}
We omit the full derivation here as, again, it conceptually follows the one for the simplified case in \cite{echeandia2021numerical}, but with extra tediousness for tracking tensor storage.

Overall, Algo.~\ref{alg:Gamma} has computational complexity $\mathcal{O}(N_c d^2)$ when applied to a fixed library of clusters. The first for loop (Lines 1-4) has complexity $\mathcal{O}(N_c)$. The main for loop (Lines 5-32) executes $N_c$ times, while the nested while loops (Lines 10-30 and 13-27) execute at most $d$ times each. The assignments on Lines 15-17 give the symmetric parts of the generalized Christoffel symbols, while the two if statements (Lines 18 and 21) cover cases where the basis vector fields do not commute, as would be signified by a quantity $\mathcal{D}_i = \frac{\partial \mathring{\gbPhi}_i }{\partial \vjb{i}}$ lacking symmetry in its 2-3 dimensions. {\sc Matlab} source code is available at \cite{wensingweb_Christoffel}.
\fi

\section{Supporting Algorithms for Adaptive Control with Local Closed Loops}
\label{app:adaptive}

This appendix presents a series of algorithms to support adaptive control. These draw heavily from past work (e.g., \cite{niemeyer1991performance} and \cite{wang2012recursive}) while noting generalizations of their treatments to the case of local loop closures modeled via constraint embedding \cite{jain2012multibody,chignoli2023propagation}. We first provide an $O(N_c)$ algorithm for direct adaptive control, while then noting a variety of $O(N_c d)$ algorithms for computing regressors that may be used for indirect or composite adaptive control.

\subsection{Direct Adaptive Control}
The direct adaptive control law of Slotine and Li \cite{slotine1987adaptive}:
\begin{align}
\tau &= Y(q,v,v_r,\dot{v}_r) \hat{\theta} - K_D s \label{eq:direct1}\\
\dot{\hat{\theta}} &= - Y(q,v,v_r,\dot{v}_r)\T s \label{eq:direct2}
\end{align}
can be computed efficiently for open-chain systems via calculations in \cite{niemeyer1991performance}. However, the same structure of these computations holds when considering closed-chain mechanisms that are modeled via constraint embedding. 
To illustrate, we note that
\[
\tau_i = [\hat{H}(q) \dot{v}_r + \hat{C}(q,v) v_r + \hat{g}(q)]_i
\]
can be implemented recursively using  
\begin{align*}
\tau_i &= \gbPhi_i\T \sum_{j\succeq i} \gXM{j}{i}\T \gF_j \\
\gF_i  &= \hat{\gI}_i \ga_i + \gB(\hat{\gI}_i, \gv_i) \gw_i \\
\ga_i  &= \gXM{i}{p(i)} \ga_{p(i)} + \gPhiDot_i(q, v) \vj{r,i} + \gbPhi_i \dot{\vj{}}_{r,i}  \\
\gv_i  &= \gXM{i}{p(i)} \gv_{p(i)} + \gbPhi_i \vj{r,i}\\
\vw_i  &= \gXM{i}{p(i)} \gw_{p(i)} + \gbPhi_i \vj{i}
\end{align*}
where $\gPhiDot_i$ is its derivative under the true velocity $v$, while $\gw_i$ is the velocity that cluster $i$ takes with reference velocities $v_r$. 

To see the recursive implementation, we consider a body-level regressor helper function:
\[
\mathcal{Y}_i(\ga_i, \gv_i, \gw_i)  = \frac{\partial}{\partial \theta_i} \left[ \gI_i \ga_i + \gB(\gI_i, \gv_i) \gw_i \right]
\]
so that $ \gF_i = \mathcal{Y}_i(\ga_i, \gv_i, \gw_i) \hat{\theta}_i$. The important result here is that the cluster-level forces remain linear in the inertia parameters of each body, extending the role that the body parameters play in the conventional case. As a result, the implementation for the certainty equivalence portion ($Y \hat{\theta}$) in \eqref{eq:direct1} follows the structure of a recursion, as shown in lines \ref{line:rnea_start}-\ref{line:rnea_end} in Alg.~\ref{alg:direct}.

Regarding the recursive implementation of the update law \eqref{eq:direct2}, we see that
\[
(Y \hat{\theta})_i = \sum_{j \succeq i} \gbPhi_i \T \gXM{j}{i}\T \mathcal{Y}_j(\ga_j, \gv_j, \gw_j) \hat{\theta}_j
\]
so that the block $Y_{ij}$ relating the torque at joint $i$ to the parameters of cluster $j$ takes the form
\[
Y_{ij} = \begin{cases} \gbPhi_i \T \gXM{j}{i}\T \mathcal{Y}_j(\ga_j, \gv_j, \gw_j)  & j \succeq i \\ 0 & o/w \end{cases}
\]
With this form, we can write the parameter update term
\[
(Y^T s)_j = \sum_i Y_{ij}\T {\tt s}_i = \sum_{i\preceq j} \mathcal{Y}_j\T \gXM{j}{i} \gbPhi_i {\tt s}_i  
\]
where ${\tt s}_i$ denotes the components of $s$ associated with cluster joint $i$. However, owing to the fact that $s = v - v_r$, the sum can be reorganized using:
\[
\sum_{i \preceq j } \gXM{j}{i} \gbPhi_i {\tt s}_{i} = \gv_j - \gw_j 
\]
Ultimately providing:
\[
(Y^T s)_j = \mathcal{Y}_j\T (\gv_j - \gw_j)
\]
which is implemented on line \ref{line:YTs} of Alg.~\ref{alg:direct}.

\begin{algorithm}[t]
\caption{ Algorithm to Support Direct Adaptive Control}
\setstretch{1.2}
\begin{algorithmic}[1]
\REQUIRE $q, v = [\vj{1};\ldots; \vj{N}], v_r, \dot{v}_r, \hat{\theta}$
\STATE $\ga_0 = - {}^0 \va_g$
\FOR{$i=1$ to $N$} \label{line:rnea_start}
\STATE $\gv_i = \gXM{i}{p(i)}\,\gv_{p(i)}+\gbPhi_i \,\vj{i}$
\STATE $\gw_i = \gXM{i}{p(i)}\,\gw_{p(i)}+\gbPhi_i \,\vj{r,i}$
\STATE $\gPhiDot_i = \mathring{\gbPhi}_i(\qj{i}, \vj{i}) + (\gv_i \times) \gbPhi_i$
\STATE $\ga_i = \gXM{i}{p(i)}\,\ga_{p(i)}+\gPhiDot_i \,\vj{r,i} +   \gbPhi_i \, \dot{\vj{}}_{r,i} $
\STATE $\hat{\gI}_i = {\rm diag}(\hat{\vI}_{i_1}, \ldots, \hat{\vI}_{i_{n_i}})$ 
\STATE $\gF_i = \hat{\gI}_i \ga_i + \gB( \hat{\gI}_i, \gv_i) \gw_i$
\STATE $[Y\T s]_i = \mathcal{Y}_i(\ga_i, \gv_i, \gw_i)\T (\gv_i - \gw_i)$ \label{line:YTs}
\ENDFOR
\FOR{$i=N$ to $1$}
\STATE $\tau_i = \gbPhi_i\T \gF_i$
\STATE $\gF_{p(i)}+= \gXM{i}{p(i)}\T \gF_i $
\ENDFOR \label{line:rnea_end}
\STATE $\tau = \tau - K_D s$ \label{line:addDamping}
\RETURN $\tau$, $Y^T s$
\end{algorithmic}
\label{alg:direct}
\end{algorithm}

\subsection{Indirect and Composite Adaptation}
In indirect \cite{li1989indirect} or composite adaptive control settings \cite{slotine1989composite} or in momentum-based estimation methods (e.g., \cite{xiu2016experimental,nedelchev2023enhanced}), it is often desirable to compute other regressors. Both sets of approaches are motivated by integrating or filtering both signals in an equation where one side is known (or measured) and the other is linear in the inertial parameters. 

\subsubsection{Indirect Adapatation Based on Momentum}
As a first instance, consider Hamilton's equations:
\begin{align*}
v &= H(q)^{-1} p \\
\dot{p} &= \tau + C(q,v)\T v - g(q) 
\end{align*}
where the generalized momentum $p = H(q) v$.
Applying a low-pass filter to both sides of the second equation and re-arranging provides
\[
\frac{\lambda}{s+\lambda} \tau = \frac{\lambda}{s+\lambda} [\dot{p} - C\T v + g]
\]
where we now consider $s$ as the Laplace variable. We can use the impulse response $w(t) = \lambda e^{-\lambda t}$ of the filter to simplify the output of the filter $ \frac{\lambda}{s+\lambda} \dot{p}$ at time $t$ via convolution
\begin{align*}
&\int_0^t w(t-t') \dot{p}(t') {\rm d}t' \\ &~~~~~~~= w(0) p(t) - w(t) p(0) -  \lambda \int_0^t w(t-t') p(t') {\rm d}t'  
\end{align*}
so that at time $t$ we have
\begin{align*}
\frac{\lambda}{s+\lambda} \tau &= \lambda p(t) - w(t) p(0) - \frac{\lambda}{s+\lambda} \left[ \lambda p + C\T v - g \right]
\end{align*}

Overall, to adapt to any residual in this equation (e.g., via \cite{li1989indirect,slotine1989composite}), we would like corresponding regressors for $p = H v$, 
$C\T v$ and $g$. We consider regressors so that 
\begin{align*}
C^T v &= Y_{c}(q,v) \, \theta \\ 
p & = Y_p(q,v) \, \theta \\
g &= Y_g(q) \,\theta 
\end{align*}

\begin{algorithm}[t!]
\caption{Regressors}
\setstretch{1.2}
\begin{algorithmic}[1]
\REQUIRE $q, v, v_r, \dot{v}_r$
\STATE $\ga_0 = \ga_{g,0} = - {}^0 \va_g$
\STATE $\gv_0 = \gw = 0$ 
\FOR{$i=1$ to $N_c$}
\STATE $\gv_i = \gXM{i}{p(i)}\,\gv_{p(i)}+\gbPhi_i \,\vj{i}$
\STATE $\gw_i = \gXM{i}{p(i)}\,\gw_{p(i)}+\gbPhi_i \,\vj{r,i}$
\STATE $\gPhiDot_i = \mathring{\gbPhi}_i(\qj{i}, \vj{i}) + (\gv_i \times) \gbPhi_i$
\STATE $\ga_i = \gXM{i}{p(i)}\,\ga_{p(i)}+\gPhiDot_i \,\vj{r,i} +   \gbPhi_i \, \dot{\vj{}}_{r,i} $
\STATE $\ga_{g,i} = \gXM{i}{p(i)}\,\ga_{g,p(i)}$
\STATE $\mathsf{F}_i = \mathcal{Y}_i(\ga_i, \gv_i, \gw_i)$ 
\STATE $\mathsf{H}_i = \mathcal{Y}_i(\gv_i, 0,0)$ 
\STATE $\mathsf{G}_{i} = \mathcal{Y}_i(\ga_{g,i}, 0,0)$ 
\STATE $[Y_T]_i= \frac{1}{2} \gv_i\T \mathsf{H}_i$ 
\STATE $[Y_{\dot{V}}]_i=  \gv_i\T \mathsf{G}_i$
\ENDFOR
\FOR{$j=N_c$ to $1$}
  \STATE $i = j$
  \WHILE{$i>0$}
    \STATE $[Y]_{ij} = \gbPhi_i\T \mathsf{F}_j$
    \STATE $[Y_p]_{ij} = \gbPhi_i\T \mathsf{H}_j$
    \STATE $[Y_c]_{ij} = \gPhiDot_i\T \mathsf{H}_j$ 
    \STATE $[Y_g]_{ij} = \gbPhi_i\T \mathsf{G}_j$
    \STATE $\mathsf{F}_j= \gXM{i}{p(i)}\T \mathsf{F}_j$ ~~~ $\mathsf{H}_j= \gXM{i}{p(i)}\T \mathsf{H}_j$ ~~~ $\mathsf{G}_j= \gXM{i}{p(i)}\T \mathsf{G}_j$
    \STATE $i = p(i)$
  \ENDWHILE
\ENDFOR
\RETURN $Y$, $Y_p$, $Y_g$, $Y_c$ , $Y_T$, $Y_{\dot{V}}$
\end{algorithmic}
\label{alg:Regressor}
\end{algorithm}

We can form a regressor for $p = H \dot{q}$ via 
\[
p = Y(q, 0,0,v ) \theta =: Y_p(q,v) \theta
\]
evaluated in the absence of gravity. Likewise the gravitational force regressor $Y_g(q) = Y(q,0,0,0)$. This leaves the evaluation of a regressor for $C\T v$ as the remaining item.

It is notable that $C^T v$ takes on the same value, no matter which admissible factorization of $C(q,v)$ is used. It can be shown that our choice of $\gB$ from before satisfies $\gB\T \gv = 0$, which, using our equation \eqref{eq:CbStar} gives, 
\[
C\T v = \gPhiDot\T \, \gX\T \, \gh
\]
where 
$\gh = \gI \gv$ 
collects the spatial momentum of each of the bodies. Thus, we can compute 
$C\T v$ 
via \cite{wang2012recursive}
\begin{equation}
(C\T v)_i = \gPhiDot_i\T \sum_{j \succeq i} \gXM{j}{i}\T \gI_j \gv_j
\label{eq:CTv}
\end{equation}
which is linear in the inertial parameters since $ \gI_j \gv_j =\mathcal{Y}_j(\gv_j, 0,0) \theta_j$. As a result, we have:
\[
[ Y_c ]_{ij} = \begin{cases} \gPhiDot_i\T \gXM{j}{i}\T \mathcal{Y}_j(\gv_j, 0,0)  & j \succeq i \\ 0 & o/w \end{cases}
\]

\subsubsection{Indirect Adaptation Based on Energy}
As a second instance, we could instead note that:
\[
\ddt \left[ T + V \right] = v\T \tau
\]
Now $\ddt V = v\T g$ while $\ddt{T}$ will necessarily have acceleration terms $\dot{v}$  that are difficult to measure. Using the same low-pass filter trick as before, we have:
\[
\frac{\lambda}{s+\lambda} v\T \tau = \lambda T(t)-w(t)T(0)  - \frac{\lambda}{s+\lambda} \left[ \lambda T - v\T g \right]
\]

Overall, to adapt to residuals in this equation, we would like corresponding regressors for $T$ and 
$\dot{V}$. Analogous to before, we consider regressors so that 
\begin{align*}
T &= Y_T(q,v) \,\theta \\ 
\dot{V} &= Y_{\dot{V}}(q,v)\, \theta = v\T g 
\end{align*}

To compute these regressors recursively, we re-exploit our cluster-level regressors $\mathcal{Y}_i$ which can be used to show: 
\begin{align*}
T & = \sum_i \frac{1}{2} \gv_i \T \gI_i \gv_i = \sum_i \frac{1}{2} \gv_i \T \mathcal{Y}_i(\gv_i, 0,0) \theta_i \\
\dot{V} & = \sum_i -\gv_i\T \gI_i \gXM{i}{0} {}^0 \va_g  =\sum_i \gv_i\T \mathcal{Y}_i( -\gXM{i}{0} {}^0 \va_g, 0 , 0 ) \theta_i 
\end{align*}
We can use these to then get sub-blocks of the regressors for the scalar quantities via:
\begin{align*}
[Y_T(q,v)]_i &= \frac{1}{2} \gv_i \T \mathcal{Y}_i(\gv_i, 0,0) \\
[Y_{\dot{V}}(q,v)]_i &= \gv_i\T \mathcal{Y}_i( -\gXM{i}{0} {}^0 \va_g, 0 , 0 )
\end{align*}


\subsubsection{Recursive Computation of Indirect Regressors} Putting these together, we can compute all of our regressors in a unified algorithm, given in Alg.~\ref{alg:Regressor}, with source code at \cite{wensingweb_adaptive}.

\vfill
\fi

\end{document}